\documentclass[11pt]{article}
\usepackage{amssymb}
\usepackage{amsmath}
\usepackage{amsthm}
\usepackage{amsfonts}
\usepackage{graphicx}
\usepackage{enumerate}
\usepackage{bbm}
\usepackage{dsfont}
\usepackage[round]{natbib}
\usepackage{xcolor}
\usepackage[hidelinks]{hyperref}
\usepackage{adjustbox}
\usepackage{changepage}

\usepackage{multirow}
\usepackage{multicol}
\usepackage{tikz}
\numberwithin{equation}{section}
\newtheorem{theorem}{Theorem}[section]
\newtheorem{algorithm}{Algorithm}[section]
\newtheorem{assumption}{Assumption}[section]

\newtheorem{example}{Example}

\newtheorem{proposition}{Proposition}
\newtheorem{remark}{Remark}[section]

\newtheorem{lemma}{Lemma}[section]
\theoremstyle{definition}

\renewcommand{\textheight}{22cm}
\DeclareMathOperator{\E}{\text{E}}

\renewcommand{\textheight}{22cm}

\renewcommand{\hat}{\widehat}
\renewcommand{\tilde}{\widetilde}

\newcommand{\En}{\mathbb{E}_n}
\newcommand{\Enk}{\mathbb{E}_{n,k}}

\newcommand{\Ep}{{\mathrm{E}_{P}}}
\usepackage{chngcntr} 
\counterwithout{table}{section}  

\counterwithout{figure}{section}

\begin{document}
\interfootnotelinepenalty=10000
\title{Identification-robust inference for the LATE with\\ high-dimensional covariates\thanks{\footnotesize\setlength{\baselineskip}{4.4mm} First arXiv date: February 20, 2023\smallskip}}

\date{\today}
\author{Yukun Ma\thanks{\footnotesize\setlength{\baselineskip}{4.4mm} Yukun Ma: yma69@ur.rochester.edu. Department of Economics, University of Rochester, 238 Harkness Hall, Rochester, NY 14627\smallskip}
\thanks{\footnotesize\setlength{\baselineskip}{4.4mm} We thank  Bertille Antoine, Harold Chiang, Jean-Marie Dufour, Atsushi Inoue, Edward Kennedy, Vadim Marmer, Anna Mikusheva, Whitney Newey, Yuya Sasaki, Takuya Ura, and seminar participants at Pennsylvania State University, Georgia Institute of Technology, Florida State University, NY Camp Econometrics XVII, the 2023 Asia Meeting of the Econometric Society, the 38th meeting of European Economic Association, the 33rd Midwest Econometrics Group, the 38th Canadian Econometrics Study Group, the 93rd Annual Meeting of the Southern Economic Association, and the IAAE 2024 Annual Conference. All remaining errors are our own.\bigskip}
}
\maketitle

\begin{abstract}\setlength{\baselineskip}{7.1mm} 
This paper presents an inference method for the local average treatment effect (LATE) in the presence of high-dimensional covariates, regardless of the strength of identification. We propose an orthogonalized Anderson-Rubin test statistic that maintains uniformly valid asymptotic size. We provide an easy-to-implement algorithm for inferring the high-dimensional LATE by inverting our test statistic and employing the double/debiased machine learning method. Simulation results show that our test achieves better size control under both weak identification and high dimensionality, outperforming conventional alternatives. Applying the proposed method to railroad and population data to study the effect of railroad access on urban population growth, we observe wider confidence intervals than those obtained using conventional methods.
\end{abstract}
\newpage
\section{Introduction}
We propose an orthogonalized Anderson-Rubin (AR) test statistic with uniformly correct asymptotic size. The proposed method exhibits robustness against weak identification and high dimensionality in the LATE framework.
Furthermore, we provide a practical guideline, including a step-by-step algorithm, for drawing inferences for the LATE with high-dimensional controls. This algorithm entails (1) inverting the proposed statistic to derive confidence intervals and (2) applying machine learning approaches to overcome the regularization bias and overfitting within the high-dimensional model. The objective is motivated by the persistent challenges of the weak-instrument problem in empirical research and the prevalence of rich data in the contemporary big data era.

In models where certain explanatory variables correlate with the error term, least squares estimators yield inconsistent coefficient estimates. To address this, instrumental variables are often employed, as they are uncorrelated with the error term but correlated with the endogenous explanatory variables. Nonetheless, if the correlation between the instruments and endogenous variables is weak, IV estimation becomes imprecise, resulting in unreliable tests and confidence intervals. This presents the weak-instrument problem, a notable concern in empirical practice.

Empirical researchers often seek to estimate the causal effects of endogenous regressors using instrumental variables regression. A prominent example is the influential study by \cite{angrist1991does}, which uses quarter of birth as an instrument to estimate the returns to schooling. However, \cite{bound1995problems} argue that Angrist and Krueger's results may be unreliable due to the weak correlation between one's quarter of birth and their education attainment.
Moreover, the common practice of pretesting, with a rule-of-thumb F-statistic threshold of 10 proposed by \cite{staiger1997instrumental}, is challenged by \cite{lee2022valid}. In their paper, they introduce a novel critical value function and reveal that achieving a true 5 percent test with a critical value of 1.96 instead requires an F exceeding 104.7. Applying this criterion to their sample of 61 \emph{American Economic Review} papers published between 2013 and 2019, they find that a quarter of the specifications initially presumed to be statistically significant are, in fact, insignificant.

\cite{imbens1994identification} introduced a widely used framework for estimating the LATE. This parameter captures the treatment effect for the subgroup of compliers who take the treatment if and only if they are assigned to the treatment group. The use of instrumental variables to estimate LATE has received considerable attention in the literature.
Within the LATE framework, weak identification emerges either when instruments correlate weakly with endogenous regressors or when the proportion of compliers is relatively small. We are particularly interested in exploring and addressing this weak identification issue within the LATE framework for several reasons. On the one hand, while compliers might be a minority, they often represent the population of critical interest to policymakers. Take the Vietnam-era draft lottery in \cite{angrist1990lifetime} as an example. Even though compliers constituted a minority, approximately 0.10 to 0.16, their experiences provide valuable insights into the draft's direct consequences for individuals at the decision-making margin, thereby shedding light on the immediate effects of veteran status on civilian earnings. On the other hand, natural experiments often present challenges of weak identification, especially when researchers have no control over the size of the complier group, as highlighted by the instrument strength concerns in \cite{angrist1991does}. Our objective is to yield reliable outcomes regardless of the proportion of individuals affected by the intervention, both theoretically and practically. Understanding this impact is crucial for expanding or scaling such interventions.

The weak-instrument literature has produced a range of econometric techniques for estimating and conducting inference on a structural parameter $\theta$ defined by moment conditions. In these models, one specifies a score $\psi(W;\theta)$ such that $\Ep[\psi(W;\theta_0)]=0$ at the true value $\theta_0$. We introduce an identification-robust test of the null hypothesis $H_0(\theta_0): \Ep[\psi(W;\theta_0)]=0$, which is equivalent to $H_0:\theta=\theta_0$. This problem has been addressed by, among others, \cite{anderson1949estimation}, \cite{stock2000gmm}, \cite{kleibergen2002pivotal}, and \cite{andrews2016conditional}. While these studies present methods tailored for inference about target parameters in the presence of weak identification, they overlook models equipped with high-dimensional covariates. Identification-robust testing procedures are particularly important in high-dimensional settings, as it is unclear how to pretest for weak identification in such cases.
On the one hand, accounting for a large number of covariates can bolster the validity of the IV within the LATE framework. On the other hand, existing identification-robust methods face challenges in the presence of many covariates, especially due to size distortions when many covariates are included.

Our proposed method addresses these challenges by accommodating any arbitrary $N/p$ ratio, where $N$ is the sample size and $p$ is the number of covariates. This feature ensures that our method maintains correct size control regardless of the covariate dimensionality. Our technical innovation diverges from the typical approach of proposing a consistent LATE estimator. Instead, we introduce a stochastic process and its uniformly consistent estimator over the probability law under the null hypothesis. Building on this, we introduce a test statistic that exhibits uniformly correct asymptotic size.
Our contribution generalizes prior research by developing an identification-robust statistic that employs machine learning techniques, enabling us to explore a broader set of controls than previously possible.

Our simulation results demonstrate that the proposed orthogonalized AR method performs reliably across identification strengths and covariate dimensionality, outperforming both the conventional AR test and the double/debiased machine learning (DML) method of \citet{CCDDHNR18}. The conventional AR test exhibits severe size inflation as dimensionality increases and becomes infeasible when the number of covariates approaches or exceeds the sample size. The DML method, while valid for high-dimensional nuisance estimation, shows size deflation and over-coverage when instruments are weak or moderately weak, resulting in under-rejection. In contrast, the proposed method maintains empirical size close to the nominal level across sample sizes, identification strengths, and dimensionality. Under strong identification, its power is comparable to that of the DML method, while under weak identification it remains stable and correctly sized. These findings indicate that the proposed procedure combines the identification-robustness of the AR framework with regularized nuisance estimation, providing reliable inference across a wide range of data environments.

We further evaluate the proposed method in two empirical applications. The first, following \citet{hornung2015railroads}, examines the effect of railroad access on city population growth in nineteenth-century Prussia. The second, following \citet{ambrus2020loss}, investigates the long-term effect of cholera-related deaths on property values in nineteenth-century London. In both applications, the proposed method yields confidence intervals that are wider than those from the conventional AR test and the DML method, which is consistent with identification-robust inference when many covariates are included and instrument strength may be limited. Within our implementation, the resulting confidence intervals and inference outcomes are stable across regularization choices, including Lasso, Ridge, and Elastic Net, and across alternative control sets. Several effects that appear statistically significant under the AR test or the DML method are not significant under the proposed method. Taken together, these findings indicate that identification-robust inference with regularized nuisance estimation may alter statistical significance assessments in settings with many covariates and uncertain instrument strength.

This paper advances the well-established literature on weak identification by providing procedures for inference and the construction of confidence intervals for the LATE parameters in high-dimensional models. We develop a test statistic with uniformly correct asymptotic size. Furthermore, we provide a practical guideline, complete with a step-by-step algorithm, for drawing inferences and determining confidence intervals for the high-dimensional LATE using machine learning methods.

\subsection{Relations to the Literature}
This paper contributes to the literature on weak identification and high-dimensional models by providing a test tailored for making inferences regarding the LATE in the presence of high-dimensional covariates.

Since \cite{staiger1997instrumental} introduced the ``local-to-zero'' framework for weak instruments, a sequence of identification-robust tests has emerged.\footnote{See works by \cite{bound1995problems},
 \cite{kleibergen2002pivotal}, 
\cite{andrews2006optimal},
\cite{moreira2009tests},
 \cite{andrews2019identification},
 \cite{moreira2019optimal},
\cite{mikusheva2022inference} for various identification-robust inference methods developed over the past three decades.
For detailed surveys on the weak identification literature, see
 \cite{stock2002survey},
 \cite{dufour2003identification},
\cite{andrews2005inference}, and \cite{andrews2019weak}.} 
To test if the mean function equals zero at the true parameter value $\theta_0$, 
\cite{stock2000gmm} pioneered the concept of weakly identified Generalized Method of Moments (GMM). They introduced the $S$ statistic as the quadratic form of the objective function, which is a generalized form of the AR test statistic as in \cite{anderson1949estimation} and follows a $\chi^2$ asymptotic distribution under the null hypothesis. Later, \cite{kleibergen2005testing} proposed the $K$ statistic, capitalizing on the asymptotic independence between the Jacobian estimator of the objective function and the sample average of the moment. \cite{moreira2003conditional} sharpened power properties by constructing a conditional likelihood-ratio test that delivers exact size and locally most-powerful unbiased inference in linear IV models. 
\cite{andrews2016conditional} further generalized the conditional approach by conditioning on the entire observed path of the sample moment process, obtaining procedures that remain valid regardless of identification strength and that achieve near-efficient power in both point-identified and set-identified designs. In the just-identified LATE setting with a single endogenous regressor and a single instrument, these statistics reduce algebraically to the AR test. We build on this equivalence and develop an orthogonalized AR procedure that retains identification-robust guarantees while accommodating high-dimensional nuisance learning through Neyman orthogonal scores and cross-fitting. This framework preserves the simplicity of AR, delivers size-correct inference under weak or strong first-stage relationships, and integrates naturally with modern control-selection methods introduced below.

Over the past decade, there has been a surge in the literature on machine learning-based econometric methods for high-dimensional models. 
\cite{belloni2015uniform} advanced a Neyman orthogonal score for a Z-estimation framework in the presence of high-dimensional nuisance parameters. Subsequently, \cite{belloni2018uniformly} constructed a confidence interval rooted in the Neyman orthogonality condition in the high-dimensional setting. In a series of contributions, Chernozhukov et al. (\citeyear{chernozhukov2013gaussian}, \citeyear{chernozhukov2016empirical}, \citeyear{chernozhukov2017central}) established the Central Limit Theorem (CLT) for high-dimensional models using the Gaussian approximation approach. \cite{belloni2014high} presented an overview of techniques for estimating and inferring in high-dimensional datasets.
\cite{CCDDHNR18} introduced the DML methodology in the i.i.d. setting. They combined the Neyman orthogonality condition\footnote{We refer readers to \cite{pfanzagl1985contributions}, \cite{bickel1993efficient}, and \cite{newey1994asymptotic} for the development of the Neyman orthogonal score.} and cross-fitting methods. Most recently, \cite{chernozhukov2022locally} outlined a general construction for the doubly robust moment function, ensuring robustness against nonparametric or high-dimensional first steps. However, none of these papers on high-dimensional models consider weak identification issues.  

This paper also relates to the literature on instrumental variables estimation of the LATE. \cite{imbens1994identification} pioneered the introduction of a simple instrumental variables estimand for the average treatment effect among compliers. Motivated by \cite{angrist1991does}, \cite{angrist1995two} broadened the LATE framework to accommodate ordered treatments, such as years of schooling. A subsequent wave of research explores the incorporation of covariates into LATE estimation, including \cite{imbens1997bayesian}, \cite{angrist2000interpretation}, \cite{hirano2000assessing}, \cite{yau2001inference}, and \cite{abadie2003semiparametric}, using either parametric or semiparametric estimation approaches.
\cite{tan2006regression} proposed a LATE estimator with robustness against the misspecification of either the propensity score model or the outcome regression model. \cite{hong2010semiparametric} derived the semiparametric efficiency bounds for conditional and unconditional LATE.
\cite{frolich2007nonparametric} and \cite{ogburn2015doubly} provided a fully nonparametric $\sqrt{N}$-consistent and efficient estimator for the LATE with confounding covariates. 
More recently, \cite{belloni2017program} presented an efficient estimator alongside reliable confidence bands for the LATE with nonparametric/high-dimensional components, using the orthogonal moment condition and machine learning method. \cite{CCDDHNR18} incorporated their proposed DML method into the LATE framework, achieving an $\sqrt{N}$-consistent estimator for the LATE in the presence of high-dimensional covariates.
 \cite{angrist2022empirical} underscored the importance of the LATE framework for causal inferences through empirical demonstrations.

Previous work on treatment effect estimation with many covariates included \cite{oprescu2019orthogonal}, who proposed the orthogonal random forest to estimate heterogeneous treatment effects under unconfoundedness while allowing for nonparametric flexibility. \citet{qiu2021inference} examined covariate-specific LATEs in high dimensional settings by imposing a parametric generalized linear model and applying variable selection and debiasing techniques to recover $\sqrt{N}$-consistency. \citet{sun2022high} developed a doubly robust LATE estimator based on regularized, calibrated, and weighted $M$-estimators. \citet{boot2024inference} studied the case of many and possibly weak instruments by introducing the saturated instrumental variable estimator, which debiased two-stage least squares in a fully interacted design. While their approach targeted a weighted, stratum-specific LATE, our method focused on the unconditional LATE under the assumption of strong monotonicity. Robustness was achieved through orthogonalization and the use of machine learning for estimating nuisance components, rather than relying on saturation. Although these methods enabled flexible, high-dimensional nuisance adjustment, their inferential guarantees depended on a well-identified first stage or correctly specified functional forms.

Empirical instrumental variables studies typically estimate the LATE conditional on a vector of pre-treatment covariates \(X\) and then average over the empirical distribution of \(X\) to obtain the unconditional LATE.  Covariate adjustment relaxes the exogeneity condition to \(Z \perp\!\!\!\perp (Y(d), D(z)) \mid X\), which makes the independence and exclusion assumptions more plausible when the instrument is not fully randomized. It also absorbs heterogeneity, reducing residual variance and improving efficiency, and can enhance instrument sharpness. The findings of \cite{kennedy2020sharp} demonstrated that a rich set of covariates predicting compliance can bolster causal identification even when instruments are weak by increasing instrument sharpness. These considerations explain the widespread practice of including many control variables in LATE analyses and motivate our high-dimensional framework, which delivers identification-robust inference uniformly in the number of controls \(p\).  Accordingly, we target the unconditional LATE while allowing $p$ to grow with, or even exceed the sample size $N$. To the best of our knowledge, this paper is the first to provide identification-robust inference for the LATE under high-dimensional covariates, without imposing assumptions on instrument strength.

\subsection{Outline}
The rest of the paper is structured as follows. Section \ref{section_overview} provides a practical guideline for implementing the proposed algorithm. Section \ref{section_theory} presents the theoretical results. Section \ref{section_simulation} reports the Monte Carlo simulation findings. Section \ref{section_application} offers two empirical illustrations. Section \ref{section_conclusion} concludes. The appendix contains all proofs of the theorems and lemmas, as well as an extension beyond the LATE framework to a general instrumental variables model.
\section{Overview}
\label{section_overview}
In this section, we provide a brief overview of our proposed method without theories. This overview serves as a concise guideline in practice. 
\subsection{Notation and Target Parameter}
\label{subsection_notation}
 Consider the standard instrumental variable setup where the researcher has access to a dataset of N i.i.d. observations, represented as 
 $\{W_i=(Y_i,D_i,Z_i,X_i')\}_{i=1}^N$. The outcome of interest for unit $i$ is denoted by $Y_i$. Let $D_i\in\{0,1\}$ be a binary indicator of the receipt of treatment for unit $i$. The instrumental variable $Z_i$ is also binary and can be interpreted, for example, as the offer of treatment. This instrument is randomly assigned conditional on the covariates.  Let $X_i=(X_{i1},\cdots,X_{ip})'$ denote the $p$-dimensional vector of observed controls for unit $i$, where $X_{ij}$ represents the value of the $j$th covariate for unit $i$. Notably, the dimensionality $p$ can be substantially greater than the available sample size, $N$.  For each unit \(i\), let \(D_i(z)\) denote the potential treatment received if the instrument takes the value \(Z_i=z\), and let \(Y_i(z,d)\) denote the potential outcome if \((Z_i,D_i)=(z,d)\). The causal parameter of interest is the LATE, $\theta_{\text{LATE}} = \Ep\!\left[Y_i(1,1)-Y_i(0,0)\mid D_i(1)>D_i(0)\right]$. Under the exclusion restriction, \(Y_i(z,d)=Y_i(d)\), the estimand simplifies to
$\theta_{\text{LATE}} = \Ep\!\left[Y_i(1)-Y_i(0)\mid D_i(1)>D_i(0)\right].$
To streamline exposition, we adopt a binary representation for scalar $D$ and scalar $Z$. However, it is important to highlight that our framework can be extended to encompass broader contexts, including scenarios with ordered treatments like years of schooling, or when dealing with vector-valued $D$ and $Z$ as encountered in general instrumental variables models.  Appendix \ref{appendix_general} develops the extension to general moment-restriction models.

Let $\{\mathcal{P}_N\}_N$ denote a sequence of probability laws associated with $\{W_i\}_i$. As the sample size $N$ grows,
our analysis allows for an increasing dimensionality of $W_i$. Here, $P=P_N\in \mathcal{P}_N$ is defined with respect to a specific sample size $N$, and $\Ep$ stands for the expected value under the law $P$. For any set $B$, its complementary set is given by $B^c=\{1,\cdots,N\}\setminus B$, and $|B|$ represents the size or cardinality of $B$. We introduce the subsample expectation operator defined as $\mathbb{E}_B[\cdot]:=\frac{1}{|B|}\sum_{i\in B}[\cdot]$.
\subsection{Anderson-Rubin-Type Neyman Orthogonal Score}
\label{subsection_score}
We model the random vector $W=(Y,D,Z,X')'$ as:
\begin{alignat}{3}
& D=m_0(Z,X)+v, \quad &&\Ep[v|Z,X]=0, \quad &&(\text{First stage}) \label{model_first_stage}
\\
& Y=g_0(Z,X)+u, \quad &&\Ep[u|Z,X]=0, \quad &&(\text{Reduced form}) \label{model_reduce_form}
\\
& Z=p_0(X)+e, \quad &&\Ep[e|X]=0, \quad &&(\text{Propensity score})
\label{model_propensity_score}
\end{alignat}
where $m_0$ is a function that maps the support of $(Z,X)$ to $(\varepsilon,1-\varepsilon)$, $g_0$ is a function that maps the support of $(Z,X)$ to $\mathbb{R}$, $p_0$ is a function that maps the support of $X$ to $(\varepsilon,1-\varepsilon)$ for some $\varepsilon\in (0,1/2)$, and $v, u, e$ are error terms. We do not impose any parametric assumptions\footnote{\cite{blandhol2022tsls} demonstrated that 2SLS specifications can only have a LATE interpretation when controlling for ``rich'' covariates in a nonparametric manner.} on the form of the $m_0$, $g_0$, and $p_0$ functions here.

The population LATE proposed by \cite{tan2006regression} is given by 
\begin{align}
\label{equation_LATE_formula}
\theta_\text{LATE}=\frac{\Ep[g_0(1,X)-g_0(0,X)]+\Ep\left[\frac{Z(Y-g_0(1,X))}{p_0(X)}\right]-\Ep\left[\frac{(1-Z)(Y-g_0(0,X))}{1-p_0(X)}\right]}{\Ep[m_0(1,X)-m_0(0,X)]+\Ep\left[\frac{Z(D-m_0(1,X))}{p_0(X)}\right]-\Ep\left[\frac{(1-Z)(D-m_0(0,X))}{1-p_0(X)}\right]}.
\end{align} 
The numerator includes the intent-to-treat (ITT) component along with augmentation terms that ensure robustness and orthogonality, while the denominator contains the compliance probability with analogous augmentation. The estimand in equation \eqref{equation_LATE_formula} coincides with the causal LATE defined in Section \ref{subsection_notation}. Under the identification conditions for LATE, the augmentation terms have mean zero, implying that the numerator simplifies to the ITT effect and the denominator to the compliance probability. Let us define the compliance probability as $B_N=\Ep[m_0(1,X)-m_0(0,X)]$. We retain the augmentation to construct a Neyman orthogonal score, which preserves the estimand but facilitates robust inference.

 The standard normal distribution of the LATE estimator can be derived using the delta method, which linearizes the LATE estimator with respect to the estimators of the numerator and denominator in equation \eqref{equation_LATE_formula}. In line with the weak-instrument literature, we model weak identification by allowing the denominator to shrink so that the concentration parameter remains bounded as stated in Remark \ref{remark_source_weak_ID}. This corresponds to a small compliance probability.
In such instances, the standard normal approximation fails because the LATE estimator is highly nonlinear with respect to the denominator estimator as the denominator approaches zero. 
To establish valid hypothesis tests and confidence sets for LATE without considering identification strength, we consider the function $\psi$ defined by 
\begin{align}
\label{equation_score_LATE}
\psi(W;\theta,\eta)
&=g(1,X)-g(0,X)+\frac{Z(Y-g(1,X))}{p(X)}-\frac{(1-Z)(Y-g(0,X))}{1-p(X)}
\\
&
\nonumber
-\theta\times\left(
m(1,X)-m(0,X)+\frac{Z(D-m(1,X))}{p(X)}-\frac{(1-Z)(D-m(0,X))}{1-p(X)}
\right),
\end{align}
where $W=(Y,D,Z,X')'$, $\theta\in\Theta$ is our target parameter LATE, with $\Theta$ being a compact set on $\mathbb{R}$, and $\eta=(g,m,p)\in\mathcal{T}$\footnote{$\mathcal{T}$ is assumed to be a convex set because we want to ensure that $\psi(W;\theta_0,\eta_0+r(\eta-\eta_0))$ is well defined. Given that $\mathcal{T}$ is convex, $\eta_0+r(\eta-\eta_0)=(1-r)\eta_0+r\eta\in \mathcal{T}$ for all $r\in[0,1)$ and $\eta\in\mathcal{T}$.} are the nuisance parameters.  
In Section \ref{section_theory}, we demonstrate that the function $\psi$ adheres to the Neyman orthogonality condition. It is important to note that the score $\psi$ satisfies the moment condition $\Ep[\psi(W;\theta_0,\eta_0)]=0$, where $\theta_0$ and $\eta_0$ are the true values of $\theta$ and $\eta$, respectively.
The orthogonal score is the same as that in Section 5.2 in \cite{CCDDHNR18}, specialized to the LATE setting. With these properties, we can describe the function $\psi$ as an Anderson-Rubin-type (AR-type) Neyman orthogonal score function for the model \eqref{model_first_stage}-\eqref{model_propensity_score}.

\begin{remark}[Source of weak identification]\label{remark_source_weak_ID}
Let \(S(\theta):=\Ep[\psi(W;\theta,\eta_0)]\). The identifying signal is the slope of this population moment with respect to \(\theta\). Evaluated at \(\theta_0\),
\[
\partial_\theta S(\theta)\big|_{\theta=\theta_0}
= -\,\Ep\!\Big[m_0(1,X)-m_0(0,X)
+ \frac{Z(D-m_0(1,X))}{p_0(X)}
- \frac{(1-Z)(D-m_0(0,X))}{1-p_0(X)}\Big].
\]
Under the models in \eqref{model_first_stage}--\eqref{model_propensity_score} and the overlap condition, the augmentation terms have mean zero by iterated expectations. Consequently,
$
-\partial_\theta S(\theta_0)=\Ep\!\big[m_0(1,X)-m_0(0,X)\big]= B_N,
$
which equals the compliance probability. Weak identification arises when \(B_N=O(N^{-1/2})\), so that \(S(\theta)\) is nearly flat in a neighborhood of \(\theta_0\).
An analogue of the linear IV concentration parameter that summarizes identification strength is
\[
\kappa_N^2 \;:=\; \frac{N\,B_N^{2}}{\Ep\!\left[\operatorname{Var}(D\mid Z,X)\right]}.
\]
Under overlap, \(\Ep[\operatorname{Var}(D\mid Z,X)]\) is bounded above and below by positive constants. Hence \(\kappa_N^2=O(1)\) if and only if \(|\partial_\theta S(\theta_0)|=O(N^{-1/2})\), which corresponds to the weak identification regime in \citet{stock2000gmm}. The AR-based procedure below maintains correct size uniformly over \(\kappa_N^2\), including this regime.
\end{remark}

\subsection{Inference Procedure}
We next introduce how to make inferences about the target parameter $\theta\in\Theta$. For each candidate value $\theta_0$, we test the moment condition $H_0(\theta_0): \Ep[\psi(W;\theta_0,\eta_0)]=0$ against $H_1(\theta_0): \Ep[\psi(W;\theta_0,\eta_0)]\neq 0$. Because the hypothesis is expressed directly in terms of the population moment, it imposes no restriction on the strength of identification. The test is valid whether the compliance probability is large, or small.

Initially, we estimate the first-stage nuisance parameters $\eta$, using some machine learning methods.
With a fixed positive integer $K>1$, we randomly partition $\{1,\cdots,N\}$ into $K$ parts, denoted as $\{I_k\}_{k=1}^K$. For each $k\in\{1,\cdots,K\}$, the nuisance parameter estimate $\hat{\eta}_k$ is computed using the subsample of those observations with index $i\in I_k^c$. Subsequently, we employ the cross-fitting/data-splitting method, as suggested by \cite{CCDDHNR18}, to compute the covariance estimator of the process $\sqrt{N}\psi(W_i;\cdot,\eta_0)$, which is expressed as 
 \begin{small}
 \begin{align}
 \label{equation_covariance_estimator}
\hat{\Omega}(\theta_1,\theta_2)=\frac{1}{N}\sum_{k=1}^K\sum_{i\in I_k}\psi(W_i;\theta_1,\hat{\eta}_k)\psi(W_i;\theta_2,\hat{\eta}_k)-\frac{1}{N^2}\sum_{k=1}^K\sum_{k'=1}^K\sum_{i\in I_k, i'\in I_{k'}}\psi(W_i;\theta_1,\hat{\eta}_k)\psi(W_{i'};\theta_2,\hat{\eta}_{k'}),
 \end{align}
  \end{small}
for $\theta_1,\theta_2\in\Theta$.
Observe that $\hat{\Omega}(\theta_1,\theta_2)$ is computed using the sample of observations with index $i\in I_k$ and this computation is repeated $K$ times. For a candidate value $\theta_0$, define the cross-fitted sample moment $\hat{q}_N(\theta_0)=\frac{1}{N}\sum_{k=1}^K\sum_{i\in I_k}\psi(W_i;\theta_0,\hat\eta_k)$. 
Since $\hat{\Omega}(\theta_0,\theta_0)$ is a scalar in the just-identified LATE setting, we form the statistic
 \begin{align}
 \label{equation_test_statistic}
 AR(\theta_0) = \frac{N\hat{q}_N(\theta_0)^2}{\hat{\Omega}(\theta_0,\theta_0)},
 \end{align}
and refer to it as the orthogonalized AR statistic.
Under the null hypothesis, $AR(\theta_0)$ is asymptotically $\chi^2_1$ uniformly over the strength of identification. We reject $H_0(\theta_0)$ at level $\alpha$ whenever $AR(\theta_0)>\chi^2_{1,1-\alpha}$, where $\chi^2_{1,1-\alpha}$ denotes the $(1-\alpha)$-quantile of the $\chi^2_1$ distribution. Inverting the point-wise tests over $\theta\in\Theta$ yields a $(1-\alpha)$ confidence set that remains valid whether the compliance probability is large or small.

We specifically examine a logit model class where a binary outcome $D_i$, denoting individual $i$'s receipt of treatment, is determined by the treatment offer, $Z_i$, and a set of $p$-dimensional covariates, $X_i$. Moreover, we employ the logit model to estimate the propensity score and conduct linear regression analysis to estimate the outcome regression. The models can be expressed as:
\begin{align*}
&\Ep[D_i|Z_i,X_i]=\Lambda(Z_i\beta_{11}^0+X_i'\beta_{12}^0),
\\
&\Ep[Z_i|X_i]=\Lambda(X_i'\gamma^0),
\\
&
\Ep[Y_i|Z_i,X_i]=Z_i\beta_{21}^0+X_i'\beta_{22}^0,
\end{align*}
where $\Lambda$ denotes the logistic CDF defined by $\Lambda(t)=\exp(t)/(1+\exp(t))$ for all $t\in\mathbb{R}$, and the true nuisance parameter vector $\eta_0=(\beta_{11}^0,\beta_{12}^0,\beta_{21}^0,\beta_{22}^0,\gamma^0)$. The log-likelihood functions for the logit model are $L_1(\beta_{11},\beta_{12})=\mathbb{E}_N[L_1(W_i;\beta_{11},\beta_{12})]$ and $L_2(\gamma)=\mathbb{E}_N[L_2(W_i;\gamma)]$, where $L_1(W_i;\beta_{11},\beta_{12})=\log(1+\exp(Z_i\beta_{11}+X_i'\beta_{12}))-D_i(Z_i\beta_{11}+X_i'\beta_{12})$ and $L_2(W_i;\gamma)=\log(1+\exp(X_i'\gamma))-Z_iX_i'\gamma$. We estimate the conditional mean functions $m_0(Z,X)$, $g_0(Z,X)$, and $p_0(X)$ introduced in \eqref{model_first_stage}--\eqref{model_propensity_score} using OLS or a logistic link function, and denote the corresponding estimators by $\hat{m}(Z,X), \hat{g}(Z,X)$, and $\hat{p}(X)$\footnote{Throughout, we use $\hat{m}(Z,X), \hat{g}(Z,X)$, and $\hat{p}(X)$ to denote the estimated nuisance functions obtained via cross-fitting. Specifically, for each fold $k$, $\hat{m}_k(Z, X)$, $\hat{g}_k(Z, X)$, and $\hat{p}_k(X)$ are estimated using the data excluding fold $k$, and applied to observations in fold $k$. For notational simplicity, we write $\hat{m}$, $\hat{g}$, and $\hat{p}$ without the fold subscript except when additional clarity is required.}.
The AR-type Neyman orthogonal score is then specified as 
\begin{small}
\begin{align}
\label{equation_LATE_score}
\psi(W_i;\theta,\eta)
&=\beta_{21}+\frac{Z_i(Y_i-\beta_{21}-X_i'\beta_{22})}{\Lambda(X_i'\gamma)}-\frac{(1-Z_i)(Y_i-X_i'\beta_{22})}{1-\Lambda(X_i'\gamma)}
\\ \nonumber
&-\theta\times\left[\Lambda(\beta_{11}+X_i'\beta_{12})-\Lambda(X_i'\beta_{12})+\frac{Z_i(D_i-\Lambda(\beta_{11}+X_i'\beta_{12}))}{\Lambda(X_i'\gamma)}-\frac{(1-Z_i)(D_i-\Lambda(X_i'\beta_{12}))}{1-\Lambda(X_i'\gamma)}\right].
\end{align}
\end{small}

Note that within the score function, the logit model can be easily replaced by other models, such as the probit model or linear probability model.
To illustrate the inference procedure, we outline a specific inference procedure in the subsequent algorithm. Although our algorithm primarily employs the lasso for illustration, other machine learning methods may be used in its place. We set penalty level $\lambda^k_1=\lambda^k_2=\lambda^k_3=1.1\sqrt{|I_k^c|}\Phi^{-1}(1-0.025/p)$ and construct penalty loading $\hat\Psi_1^k, \hat\Psi_2^k,$ and $\hat\Psi_3^k$ based on Algorithm 6.1 from \cite{belloni2017program}. Formally and theoretically justified choices of penalty loadings and penalty levels are elaborated in Algorithm \ref{algorithm_penalty_loading} and Lemma \ref{lemma_convergence_rate_logistic} in Appendix \ref{section_useful_lemmas}.
\begin{algorithm} (K-fold DML for high-dimensional LATE with Lasso) \label{algorithm_lasso}
\\
Step 1. Randomly split the sample with size $N$ into $K$ folds $\{I_k\}_{k=1}^K$.
\\
Step 2. 
 For each $k\in\{1,\cdots,K\}$, obtain the nuisance parameter estimates by using only the subsample of observations with indices $i\in\{1,\cdots,N\}\setminus I_k$,
\begin{enumerate}[(a)]
\item fit a partially penalized logistic regression: estimate $\hat{\beta}_{11,k}$ (unpenalized, for $Z$) and $\hat{\beta}_{12,k}$ (lasso, for $X$) of the nuisance parameters in the first-stage regression, 
\begin{align*}
&(\hat{\beta}_{11,k},\hat{\beta}_{12,k})\in\arg\min_{\beta_{11},\beta_{12}}\mathbb{E}_{I_k^c}[L_1(W_i;\beta_{11},\beta_{12})]+\frac{\lambda_1^k}{|I_k^c|}\|\hat{\Psi}^k_1\beta_{12}\|_1.
\end{align*}
\item obtain a lasso logistic estimate $\hat{\gamma}_k$ of the nuisance parameters in the propensity score model,
\begin{align*}
&\hat{\gamma}_k\in\arg\min_{\gamma}\mathbb{E}_{I_k^c}[L_2(W_i;\gamma)]+\frac{\lambda_2^k}{|I_k^c|}\|\hat{\Psi}^k_2\gamma\|_1.
\end{align*}
\item fit a partially penalized OLS regression: estimate $\hat{\beta}_{21,k}$ (unpenalized, for $Z$) and $\hat{\beta}_{22,k}$ (lasso, for $X$) of the nuisance parameters in the reduced form regression,
\begin{align*}
&(\hat{\beta}_{21,k},\hat{\beta}_{22,k})\in\arg\min_{\beta_{21},\beta_{22}}\mathbb{E}_{I_k^c}[(Y_{i}-Z_i\beta_{21}-X_i'\beta_{22})^2]+\frac{\lambda_3^k}{|I_k^c|}\|\hat{\Psi}^k_3\beta_{22}\|_1.
\end{align*}
\end{enumerate} 
Step 3. For each candidate value $\theta_0$, compute $\hat{q}_N(\theta_0)$ and $\hat{\Omega}(\theta_0,\theta_0)$ where $\hat{\Omega}$ is defined in equation (\ref{equation_covariance_estimator}) with $\hat{\eta}_k=(\hat{\beta}_{11,k},\hat{\beta}_{12,k},\hat{\beta}_{21,k},\hat{\beta}_{22,k},\hat{\gamma}_k)$ and $\psi(W;\theta,\eta)$ is defined in equation (\ref{equation_LATE_score}).
\\
Step 4. For each $\theta_0$, compute the orthogonalized AR statistic $AR(\theta_0)=N\hat{q}_N(\theta_0)^2/\hat{\Omega}(\theta_0,\theta_0)$. Reject the null hypothesis $H_0(\theta_0): \Ep[\psi(W;\theta_0,\eta_0)]=0$ whenever $AR(\theta_0)>\chi^2_{1,1-\alpha}$. The $(1-\alpha)$ confidence interval is $CI_\alpha=\{\theta\in\Theta:AR(\theta)\leq \chi^2_{1,1-\alpha}\}$.
\\
\end{algorithm}
\begin{remark}
In the partially penalized regressions in Steps 2(a) and 2(c) of Algorithm \ref{algorithm_lasso}, the coefficients on the instrument $Z$ are always included in the model and are never subject to lasso penalization. Only the coefficients on the high-dimensional controls $X$ are penalized.
\end{remark}

 \section{Theory}
 \label{section_theory}
Section \ref{section_overview} describes the algorithmic implementation of our orthogonalized AR test. We now establish its theoretical guarantees. First, we restate the score function and define the empirical process. Next, we present the  regularity conditions required for the analysis. Finally, we prove two main results: a uniform functional central limit theorem and a uniform size result for the test statistic.

The test is built on the score $\psi(W;\theta,\eta)$ in equation \eqref{equation_score_LATE}. At the true parameter value $(\theta_0,\eta_0)$, the score has zero mean and is Neyman orthogonal, which means the  G\^ateaux derivative with respect to $\eta$ vanishes.  Under these properties, estimation error in the high-dimensional nuisance functions contributes only an $o_p(N^{-1/2})$ remainder to the statistic. Appendix \ref{section_proofs} provides complete proofs.
 
For $\theta\in\Theta$, define $q_N(\theta)=N^{-1}\sum_{i=1}^N\psi(W_i;\theta,\eta_0)$ and let $S_N(\cdot)$ denote its expected value, given as $S_N(\cdot)=\Ep[q_N(\cdot)]$. With this notation, we now define an empirical process $\mathbb{G}_N(\cdot)$ as 
\begin{align*}
\mathbb{G}_N(\cdot)=\sqrt{N}( q_N(\cdot)-S_N(\cdot))=
\frac{1}{\sqrt{N}}\sum_{i=1}^N\left\{ \psi(W_i;\cdot,\eta_0)-\Ep[\psi(W;\cdot,\eta_0)]\right\}.
\end{align*}
Later we show that under mild conditions, the process $\mathbb{G}_N(\cdot)$ weakly converges to a mean-zero Gaussian process $\mathbb{G}(\cdot)$ with a covariance function $\Omega(\theta_1,\theta_2)=\Ep[\mathbb{G}(\theta_1)\mathbb{G}(\theta_2)]$. In practice $\eta_0$ is replaced by the cross-fitted estimator $\hat{\eta}_k$ from Algorithm \ref{algorithm_lasso}.
We propose an estimator for $\mathbb{G}_N(\theta)$ as
\begin{align*}
\hat{\mathbb{G}}_N(\theta)=\sqrt{N}\left\{\frac{1}{N}\sum_{k=1}^K\sum_{i\in I_k}\psi(W_i;\theta,\hat\eta_k)-\Ep\left[\psi(W_i;\theta,\hat\eta_k)\right]\right\},
\end{align*}
where the second term is a population expectation used solely for theoretical centring and is not evaluated in practice.
 An estimator of $q_N(\theta)$ is proposed as $\hat{q}_N(\theta)=N^{-1}\sum_{k=1}^K\sum_{i\in I_k}\psi(W_i;\theta,\hat\eta_k)$.  An estimator of the covariance function is
 \begin{small}
 \begin{align}
\hat{\Omega}(\theta_1,\theta_2)=\frac{1}{N}\sum_{k=1}^K\sum_{i\in I_k}\psi(W_i;\theta_1,\hat{\eta}_k)\psi(W_i;\theta_2,\hat{\eta}_k)-\frac{1}{N^2}\sum_{k=1}^K\sum_{k'=1}^K\sum_{i\in I_k, i'\in I_{k'}}\psi(W_i;\theta_1,\hat{\eta}_k)\psi(W_{i'};\theta_2,\hat{\eta}_{k'}).
 \end{align}
  \end{small}

Next we present the theoretical foundation for the orthogonalized AR test. The argument follows the empirical process framework of \cite{CCDDHNR18} but extends it to cover weak identification in the LATE setting. Whereas \cite{CCDDHNR18} establish the asymptotic normality of a $\sqrt{N}$-consistent estimator of $\theta$, we show the weak convergence result of our proposed empirical process under the null. We demonstrate that our test controls size uniformly over both weak- and strong-identification regimes.

To fix notation, for any finite-dimensional vector $\delta$, we write 
$\|\delta\|_1$ for the $\ell_1$ norm, 
$\|\delta\|_\infty$ for the sup norm, 
and $\|\delta\|_0$ for the number of nonzero components of $\delta$. 
For a function $f$, define the $L_q(P)$ norm by $\|f\|_{P,q} = (\Ep[f(W)^q])^{1/q}$.
We define the sample expectation operator as $\mathbb{E}_N[\cdot]=\frac{1}{N}\sum_{i=1}^N[\cdot]$. The prediction norm of $\delta$ is given by $\|x_{ij}'\delta\|_{2,N}=\sqrt{\mathbb{E}_N[(x_{ij}'\delta)^2]}$.  
Let us define $\mathcal{T}_{N(i)}$ as the parameter space of the $i$-th parameter in $\eta=(\eta_1,\eta_2,\eta_3)$ with $i\in\{1,2,3\}$. The sequence $\{s_N\}_{N\geq 1}$ is a set of positive integers greater than 1. Let $c_0, c_1, C_1$ be some finite and positive constants. Let $a_N=p\vee N$.
Let the sequence $\{M_N\}_{N\geq 1}$ be a set of positive constants such that $M_N\geq (\Ep[(Z_i\vee\|X_i\|_\infty)^{2\omega}])^{1/2\omega}$, where $\omega$ is a positive constant with $\omega>4$.
Let $\{\Delta_N\}_{N\geq 1}$ be a sequence of positive constants that converges to zero at a speed at most polynomial in $N$.
 For any $T\subset[p+1]$, $\delta=(\delta_1,\cdots,\delta_{p+1})'\in\mathbb{R}^{p+1}$ with $\delta_{T,j}=\delta_j$ if $j\in T$ and $\delta_{T,j}=0$ if $j\notin T$. Define the minimum and maximum sparse eigenvalue by 
\begin{align*}
\phi_{\min}(m)=\inf_{\|\delta\|_0\leq m}\frac{\|(Z_i,X_i')\delta\|_{2,N}}{\|\delta_T\|_1},\quad \phi_{\max}(m)=\sup_{\|\delta\|_0\leq m}\frac{\|(Z_i,X_i')\delta\|_{2,N}}{\|\delta_T\|_1}.
\end{align*}
Throughout, $B_N=\Ep[m_0(1,X)-m_0(0,X)]$ is the compliance probability and $\kappa_N^2 = NB_N^2/\Ep[\text{Var}(D|Z,X)]$ is the concentration parameter analogue introduced in Section \ref{subsection_score}. Recall that $\mathcal{P}_0$ denotes the collection of probability laws under the null.

\begin{assumption}{(Identification assumption for LATE)}
\label{a:identification_LATE}
\begin{enumerate}[(i)]
\item Independence:   $(Y(z,d), D(z))\perp\!\!\!\perp Z|X$ for $d,z\in\{0,1\}$.
\item Exclusion: $Y(1,d)=Y(0,d)$ for $d\in\{0,1\}$.
\item Monotonicity: $D(1)\geq D(0)$ almost surely.
\item Strong overlap: For some $\varepsilon>0$, $\varepsilon\leq P(Z=1|X)\leq 1-\varepsilon$ almost surely.
\item  Relevance/identification strength: $B_N>0$ for all $N$ and either $\sqrt{N}B_N=O(1)$ or $\sqrt{N}B_N\rightarrow\infty$.
\end{enumerate}
\end{assumption}

\begin{remark}
\cite{CCDDHNR18} also adapted their results for the LATE framework and provided regularity conditions for LATE estimation. However, their Assumption 5.2 (d) requires the denominator in \eqref{equation_LATE_formula} to be bounded away from zero, which rules out the weakly identified situation.
In contrast, Assumption~\ref{a:identification_LATE}(v) allows for a ``local-to-zero'' first stage, meaning that the average difference in treatment propensity, $\Ep[m_0(1,X) - m_0(0,X)]$, may diminish with the sample size $N$. In our notation, weak identification arises when $\sqrt{N} \lvert B_N \rvert$ remains bounded, which corresponds to a bounded concentration parameter, $\kappa_N^2 = O(1)$. Strong identification refers to the case where $\sqrt{N} \lvert B_N \rvert$ diverges, implying that $\kappa_N^2$ tends to infinity. The proposed AR-based procedure maintains correct size uniformly across both regimes.
\end{remark}
Assumption \ref{a:identification_LATE} establishes the identification assumption for the LATE framework. Assumption \ref{a:identification_LATE}(i) requires that the instrument $Z$ is independent of potential outcome and potential treatment conditional on $X$, ensuring quasi-random assignment. Assumption \ref{a:identification_LATE}(ii) stipulates that the instrument $Z$ affects the outcome solely through its impact on the treatment, with no direct effect on the outcome itself. Assumption \ref{a:identification_LATE}(iii) rules out the presence of defiers by assuming the instrument cannot decrease the probability of treatment for any unit. Assumption \ref{a:identification_LATE}(iv) is a standard overlap condition, indicating that for every value of the covariates $X$, there is a non-zero probability that a unit will either be treated or remain untreated. 

Next assumption is to guarantee good performance of the target estimator under the linear and logistic link functions.
\begin{assumption} 
\label{a:regularity_LATE}
For $P\in\mathcal{P}_N$, the following conditions hold.
\begin{enumerate}[(i)]
\item The model \eqref{model_first_stage}-\eqref{model_propensity_score} is sparse with sparsity index $\|\beta_{12}^0\|_0+\|\beta_{22}^0\|_0+\|\gamma^0\|_0\leq s_N$ and the growth restriction $N^{-1/3}\log(a_N)\leq \Delta_N$.
\item The sparse eigenvalue conditions hold with probability $1-o(1)$, namely, for some $l_N\rightarrow\infty$ slow enough, we have 
$
1\lesssim \phi_{\min}(l_Ns_N)\leq \phi_{\max}(l_Ns_N)\lesssim 1.
$
\item The moments of the models are boundedly heteroscedastic, namely $c_0\leq \Ep[u^2|X]\leq c_1$ almost surely and $\max_{j\leq p}\big\{\Ep[|X_{ij}u_i|^3]+\Ep[|X_{ij}Y_i|^3]+\Ep[|X_{ij}|^3]\big\}\leq C_1$.
\item The approximation error and empirical error obey the following boundedness and empirical regularity condition: (a) $c_0\leq \Ep[X_{ij}^2]\leq c_1$ and $N^{-1}M_n^2s_N^2\log(a_N)\leq \Delta_N$. (b) With probability $1-o(1)$, $\|\hat{g}(Z,X)-g_0(Z,X)\|_{P,2}\vee \|\hat{p}(X)-p_0(X)\|_{P,2}\vee \|\hat{m}(Z,X)-m_0(Z,X)\|_{P,2}
\leq c_1 \sqrt{s_N\log(a_N)/N}$; $\max_{j\leq p}\big\{|(\En-\Ep)[X_{ij}^2u_i^2]|\vee|(\En-\Ep)[X_{ij}^2Y_i^2]|\vee|(\En-\Ep)[X_{ij}^2v_i^2]|\vee|(\En-\Ep)[X_{ij}^2e_i^2]|\big\}\leq \Delta_N$. 
\item $\Theta$ is compact.
\item $\|Y\|_{P,q}\leq C_1$. 
\end{enumerate} 
\end{assumption}

Assumption \ref{a:regularity_LATE}(i) specifies that the number of nonzero components in the high-dimensional nuisance parameter vector is controlled by the sparsity index $s_N$, which is further bounded in part (iv)(a).
Assumption \ref{a:regularity_LATE}(ii) imposes the sparse eigenvalue condition, analogous to the RE condition in \cite{bickel2009simultaneous}.
Assumption \ref{a:regularity_LATE}(iii) imposes constraints on the error term $u$ from the reduced form equation \eqref{model_reduce_form}. Specifically, it establishes lower and upper bounds on the conditional second moment of $u$, as well as higher moment conditions on the covariates and outcomes.
Assumption \ref{a:regularity_LATE}(iv) imposes regularity conditions on both the approximation errors and the empirical errors, in a manner analogous to Assumptions 6.1 and 6.2 of \cite{belloni2017program}.

These conditions are sufficient for the high-level conditions invoked in Appendix \ref{Appendix_subsection_general_theory}.

\begin{theorem}
\label{theorem_weak_convergence}
Suppose Assumptions \ref{a:identification_LATE} and \ref{a:regularity_LATE} hold. With $\psi(W;\theta,\eta)$ defined as in equation (\ref{equation_LATE_score}), we have that the process $\hat{\mathbb{G}}_N(\cdot)$ weakly converges to a centered Gaussian process $\mathbb{G}(\cdot)$ uniformly for all $P\in\mathcal{P}_0$ with covariance function $\Omega(\theta_1,\theta_2)=\Ep[(\psi(W;\theta_1,\eta_0)-\Ep[\psi(W;\theta_1,\eta_0)])(\psi(W;\theta_2,\eta_0)-\Ep[\psi(W;\theta_2,\eta_0)])]$ as $N$ goes to infinity. The covariance function estimator $\hat{\Omega}(\theta_1,\theta_2)$ defined in \eqref{equation_covariance_estimator} concentrates around the covariance function $\Omega(\theta_1,\theta_2)$ uniformly for all $P\in\mathcal{P}_0$, in that  for any $\varepsilon>0$,
\begin{align*}
\lim_{N\rightarrow\infty}\sup_{P\in\mathcal{P}_0}P\big(\sup_{\theta_1,\theta_2}|\hat{\Omega}(\theta_1,\theta_2)-\Omega(\theta_1,\theta_2)|>\varepsilon\big)=0.
\end{align*}
\end{theorem}
\begin{proof}
See Appendix \ref{subsection_proof_thm2}.
\end{proof}
\begin{remark}
In Theorem \ref{theorem_weak_convergence}, we show that our variance estimator is a uniformly consistent estimator of $\Omega(\theta_1,\theta_2)$ across all probability laws $P\in\mathcal{P}_0$. Next, we establish that our proposed orthogonalized AR test exhibits uniformly correct asymptotic size, as formally stated in the following theorem.
\end{remark}
\begin{theorem}[Uniform size]\label{thm_correct_size}
Suppose Assumptions \ref{a:identification_LATE} and \ref{a:regularity_LATE} hold. 
For any fixed $\theta_{0}\in\Theta$, consider the test of $H_0(\theta_0):\Ep[\psi(W;\theta_0,\eta_0)]=0$ that rejects when $AR(\theta_0)>\chi^2_{1,1-\alpha}$. 
Then the test has correct asymptotic size uniformly over all data-generating processes satisfying $H_0(\theta_0)$,
\[
\lim_{N\to\infty}\;\sup_{P\in\mathcal{P}_0}
P\!\big( AR(\theta_0)>\chi^2_{1,1-\alpha} \big)=\alpha.
\]
\end{theorem}
\begin{proof}
See Appendix \ref{subsection_proof_thm3}.
\end{proof}

\citet{CCDDHNR18} establish $\sqrt{N}$-consistency and Wald-type inference for their debiased estimator under strong identification, which requires the smallest singular value of the Jacobian to be uniformly bounded away from zero. Assumption \ref{a:identification_LATE}(v) relaxes this condition by allowing the concentration parameter to remain bounded, equivalently $\sqrt{N}B_N=O(1)$. By combining the Neyman orthogonal score with the AR statistic, we obtain weak convergence of the test under the null uniformly in instrument strength, achieving size-correct inference whether the first stage is strong or arbitrarily weak. 

This shift from identification-dependent estimation to identification-robust inference yields a broader insight: by combining orthogonalization with AR test inversion, inference validity can be decoupled from first-stage strength, ensuring uniformly valid hypothesis testing even when consistent estimation is infeasible.
In high-dimensional settings, conventional diagnostics for instrument relevance can be unreliable, as the inclusion or regularized selection of many covariates effectively partials out the identifying variation. As shown by \citet{belloni2014inference}, such over-partialling can weaken the effective first stage and distort standard IV inference, highlighting the importance of orthogonalization-based procedures that safeguard inference when identification strength is uncertain.

\begin{proposition}[Asymptotic power]\label{prop_power}
Suppose Assumptions \ref{a:identification_LATE} and \ref{a:regularity_LATE} hold.
\begin{enumerate}[(i)]
\item If $\kappa_N^2=O(1)$, equivalently $\sqrt{N}B_N\to\kappa\in(0,\infty)$, and the data are generated under the fixed alternative $\theta=\theta_0+\Delta$ with fixed $\Delta\neq0$, then
\[
AR(\theta_0)\xrightarrow{d}\chi^2_1(\Xi), 
\qquad \text{with}\quad
\Xi=\frac{\nu^2}{\Omega(\theta_0,\theta_0)}, 
\quad 
\nu=-\kappa\Delta.
\]
\item If $\kappa_N^2\to\infty$, equivalently $B_N\to B>0$, and the data are generated under the local alternative $\theta=\theta_0+h/\sqrt{N}$ with fixed $h$, then
\[
AR(\theta_0)\xrightarrow{d}\chi^2_1(\Upsilon),
\qquad \text{with}\quad
\Upsilon=\frac{B^2h^2}{\Omega(\theta_0,\theta_0)}
=:I_{\text{eff}}h^2,
\]
where $I_{\text{eff}}=B^2/\Omega(\theta_0,\theta_0)$ is the semiparametric information bound for $\theta$ based on the efficient orthogonal score.
\end{enumerate}
\end{proposition}

\begin{proof}
See Appendix \ref{proof_proposition_power}.
\end{proof}

\section{Simulation Studies}

This section describes the DGP, the simulation scenarios, and the finite-sample performance of the proposed orthogonalized AR test, in comparison with conventional alternatives. 
\label{section_simulation}
\subsection{Simulation Setup}

For each replication we draw an i.i.d. sample $\{(Y_i,D_i,Z_i,X_i')\}_{i=1}^{N}$ with sample sizes $N\in\{50,100\}$ and covariate dimensions $p\in\{5,10,25,35,50,100\}$. The covariates are generated as $X_i\sim \mathcal{N}(0,\Sigma)$, where $\Sigma_{jk}=0.5^{|j-k|}$. The binary instrument is defined by 
\begin{align*}
   Z_i=\mathbbm 1\!\bigl\{\gamma_0+X_i'\gamma+\rho_{Z} h(X_i)+\nu_i\ge 0\bigr\}, \qquad \nu_i\sim\mathrm{Logistic}(0,1),
\end{align*} 
with
$\gamma_0=-0.08$, $\gamma=(0.5,0.5^{2},\dots,0.5^p)'$, and $h(X_i)=X_{i1}^2-1$.  The parameter $\rho_Z$ controls the degree of nonlinear misspecification in the instrument, and the intercept $\gamma_0$ is set so that the instrument is balanced, with equal probabilities of taking values zero and one. Instrument strength varies smoothly through a latent compliance model, with the complier share defined as
$P_C=\kappa/\sqrt{N}$ with $\kappa\in\{1.5,3,4.5,6\}$. The shares of  always-takers and never-takers are set to $P_{AT}=P_{NT}=(1-P_C)/2$.  The implied concentration parameters range from $\kappa_N^2\in\{1.88,8.87,27.23,102.86\}$. We draw $\eta_i\sim\mathrm{Unif}(0,1)$. Potential treatments are given by $D_i(0)=\mathbbm 1\{\eta_i<P_{AT}\}$ and $D_i(1)=\mathbbm 1\{\eta_i<1-P_{NT}\}$, so the realized treatment is $D_i=(1-Z_i)D_i(0)+Z_iD_i(1)$. The outcome is generated as
\begin{align*}
Y_i=\theta_0 D_i + X_i'\beta + \rho_Yg(X_i)+u_i, \qquad u_i\sim\mathcal N(0,\sigma_i^2),
\end{align*}
with $\theta_0=1$, $\beta=(0.5,0.5^{2},\dots,0.5^p)'$, and $g(X_i)=0.5X_{i1}^2+\sin(X_{i2})$. The disturbance is heteroskedastic with $\sigma_i^2=(1+\rho_\sigma|X_{i1}|)^2$. Thus, $\rho_Y$ controls outcome nonlinearity, $\rho_\sigma$ governs heteroskedasticity, and $\rho_Z$ introduces nonlinear dependence of the instrument on covariates. Setting $(\rho_Y,\rho_\sigma,\rho_Z)=(0,0,0)$ recovers the baseline linear homoskedastic design. Positive values of $(\rho_Y,\rho_\sigma,\rho_Z)$ generate misspecification and heterogeneity to assess robustness.

\subsection{Results}

Monte Carlo simulations are conducted with 3,000 iterations for each set.
We compare three approaches in the simulation study: the conventional AR test, the DML method by \cite{CCDDHNR18}, and the proposed orthogonalized AR (OAR) test. OAR and DML procedures use $K=5$ folds for cross-fitting.
For the AR test, the nuisance functions, namely the propensity score, the first-stage, and the reduced-form, are estimated using logistic regression or OLS with all covariates included without model selection or regularization. This classical low-dimensional implementation is feasible only when the design matrix has full column rank and positive residual degrees of freedom. To satisfy these conditions, in the nominal $p\in\{50,100\}$ designs we set $p=48$ for $N=50$ and $p=98$ for $N=100$, ensuring that $X'X$ is nonsingular.

In high-dimensional designs, conventional regression either fails or produces severe overfitting, so the AR test cannot be applied. In particular, no AR results are reported for $N=50$ with $p=100$. By contrast, the DML method employs Lasso-regularized regression for nuisance estimation, which enables valid inference even when $p \gg N$. Confidence intervals are constructed using a Wald-type approach based on the asymptotic normality of the debiased estimator. Our proposed OAR method combines the identification-robust test inversion of the AR procedure with regularized nuisance estimation, extending AR inference to high-dimensional settings.

Table \ref{table_N50_size} reports empirical size across identification strength, sample size, and dimensionality under the baseline linear and homoskedastic design $(\rho_Y,\rho_{\sigma},\rho_Z)=(0,0,0)$. The AR test shows pronounced size inflation that worsens as $p$ increases. Even at $N=50$ and $p=5$ the empirical size exceeds 5\% and it rises sharply with $p$, which indicates high sensitivity to dimensionality. By contrast, the DML method shows size deflation under weak or moderately weak instruments at $\kappa_N^2 \in \{1.88, 8.87\}$ with rejection rates well below the 5\% nominal level. It approaches nominal size only when instrument strength increases to $\kappa_N^2 \in \{27.23, 102.86\}$. The proposed OAR test stays near the nominal level across identification strengths and values of $p$, with only minor drift in the highest-dimensional designs. Overall, Table \ref{table_N50_size} shows that OAR is stable in both identification strength and dimensionality, while AR over-rejects increasingly with dimensionality and the DML method under-rejects when instruments are weak.

Figures \ref{figure_N50_power} and \ref{figure_N100_power} plot power curves for OAR, AR, and DML under the baseline linear and homoskedastic design $(\rho_Y,\rho_{\sigma},\rho_Z)=(0,0,0)$. The AR test shows size inflation that worsens as $p$ increases. At $\theta=1$ the size lies well above the 5\% nominal level across identification strengths. The DML method attains high power under strong identification, as seen in Figure \ref{figure_N100_power}, but displays size deflation under weak or intermediate identification, as seen in Figure \ref{figure_N50_power}, with size far below 5\%. In contrast, OAR tracks the nominal size closely across $\kappa_N^2$ and $p$. Under strong identification its power is essentially on par with DML, while under weak identification it maintains size and improves on AR.

Table \ref{table_robust} compares performance across instrument strengths $\kappa_N^2 \in \{1.88, 8.87, 27.23, 102.86\}$ for $N=50$, varying the dimensionality $p$ and considering three designs: (i) the baseline linear and homoskedastic specification $(\rho_Y,\rho_{\sigma},\rho_Z)=(0,0,0)$, (ii) outcome misspecification with heteroskedastic errors $(1,0.5,0)$, and (iii) instrument nonlinearity $(0,0,0.5)$. The conventional AR test exhibits pronounced size inflation that becomes more severe as $p$ increases, with the largest distortions under design (iii). The DML method performs well only under strong identification with $\kappa_N^2 \in \{27.23, 102.86\}$, whereas under weak or intermediate identification it displays systematic size deflation across all three designs, which indicates sensitivity to weak instruments despite tolerance to misspecification. In contrast, the proposed OAR test maintains near-nominal size across identification strengths, dimensionalities, and designs. Under strong identification, differences between OAR and DML are negligible. Overall, the simulations indicate that OAR is robust to misspecification, dimensionality, and instrument strength, that the AR test is highly sensitive to dimensionality and misspecification, and that the DML method, while resilient to misspecification, fails to deliver correct size under weak identification.

\begin{table}[htbp]
\centering
\renewcommand{\arraystretch}{0.75}
\setlength{\tabcolsep}{3pt}
\resizebox{\textwidth}{!}{%
\begin{tabular}{cccc|ccc|cccc|ccc}
\hline\hline
\multicolumn{4}{c|}{DGP} & \multicolumn{3}{c|}{Testing Procedure} &
\multicolumn{4}{c|}{DGP} & \multicolumn{3}{c}{Testing Procedure} \\
\hline
$\kappa_N^2$ & $P_C$ & $N$ & $p$ & OAR & AR & DML &
$\kappa_N^2$ & $P_C$ & $N$ & $p$ & OAR & AR & DML \\
\hline
1.88 & 0.21 & 50  & 5   & 0.051 &0.097  & 0.008 &   1.88 &  0.15 &  100 & 5   & 0.068 & 0.071 &0.005  \\
 &   &   & 10  & 0.064 &0.1395  &  0.006&    &   &   & 10  & 0.069 &0.083  &0.005  \\
&  &   & 25  &0.062  &0.584  & 0.008 &    &  & & 25  & 0.073 &  0.125&  0.004\\
 &   &   & 35  & 0.063 & 0.783 & 0.009 &    &   &   & 50  & 0.072 &  0.545&  0.005\\
 &   &   & 50  & 0.064 & 0.982 & 0.006 &    &   &   & 75  &  0.067& 0.812 &  0.005\\
 &   &   & 100 & 0.065 & $\times$ &0.007  &    &   &   & 100 &  0.066&  0.986 &  0.004\\
\hline
8.87 & 0.42  &50 & 5   &0.070  & 0.097 & 0.028 &     8.87 &  0.30 &100 & 5   & 0.072 &0.072  & 0.017 \\
 &   &   & 10  & 0.067 &0.144  & 0.026 &    &   &   & 10  & 0.068 &0.082  &  0.019\\
 &  &   & 25  & 0.066 & 0.571 & 0.025 &    &  & & 25  &  0.079&0.124  &  0.019\\
 &   &   & 35  & 0.067 &0.782  & 0.023 &    &   &   & 50  &  0.078&0.528  &  0.021\\
 &   &   & 50  & 0.075 &0.977  & 0.029 &    &   &   & 75  & 0.071 & 0.813 &  0.023\\
 &   &   & 100 & 0.710 & $\times$ &0.021  &    &   &   & 100 & 0.071 &  0.986&  0.021\\
\hline
27.23 & 0.63 &50 & 5     & 0.058 &0.094  &0.052   & 27.23&   0.45&  100 & 5   & 0.078 & 0.071 &  0.046\\
 &   &   & 10  & 0.062 &  0.142&  0.050&    &   &   & 10  &  0.069&0.084  &  0.049\\
 &  &   & 25  & 0.060 & 0.564 & 0.051 &    &  &  & 25  & 0.071 &  0.125& 0.044 \\
 &   &   & 35  & 0.062 &0.776  &  0.045&    &   &   & 50  & 0.075 &  0.521& 0.048 \\
 &   &   & 50  & 0.065 & 0.983 & 0.054 &    &   &   & 75  & 0.068 & 0.796 &  0.047\\
 &   &   & 100 &0.065  & $\times$ &0.047  &    &   &   & 100 & 0.067 &0.987  &  0.052\\
\hline
102.86 & 0.84   &50& 5   &  0.065&0.091  &0.064  &     102.86&   0.60&100& 5   & 0.075 &0.072  &  0.058\\
 &   &   & 10  & 0.072 & 0.133 &0.061  &    &   &   & 10  & 0.076 &0.083  &  0.061\\
 &  &   & 25  & 0.055 & 0.569 & 0.064 &    &  &  & 25  & 0.078 &0.126  &  0.058\\
 &   &   & 35  & 0.069 &0.748  & 0.064 &    &   &   & 50  & 0.072 & 0.519 &  0.062\\
 &   &   & 50  & 0.074 & 0.978 & 0.067 &    &   &   & 75  &  0.072&0.806  &  0.058\\
 &   &   & 100 & 0.073 & $\times$ &0.068  &    &   &   & 100 & 0.076 &0.988  &  0.063\\
\hline\hline
\end{tabular}
}
\caption{Empirical size by identification strength, sample size, and dimensionality under the baseline linear–homoskedastic design with $(\rho_Y,\rho_{\sigma},\rho_Z)=(0,0,0)$. Reported are rejection rates at the nominal 5\% level for the proposed OAR, the conventional AR, and the DML procedure. Identification strength is indexed by $\kappa_N^2$ with implied complier share $P_C$. $N$ is sample size and $p$ is the number of covariates. Each design is simulated 3,000 times. An ``$\times$'' indicates a statistic is not available.}
\label{table_N50_size}
\end{table}



\begin{figure}
\centering
\includegraphics[scale=0.6]{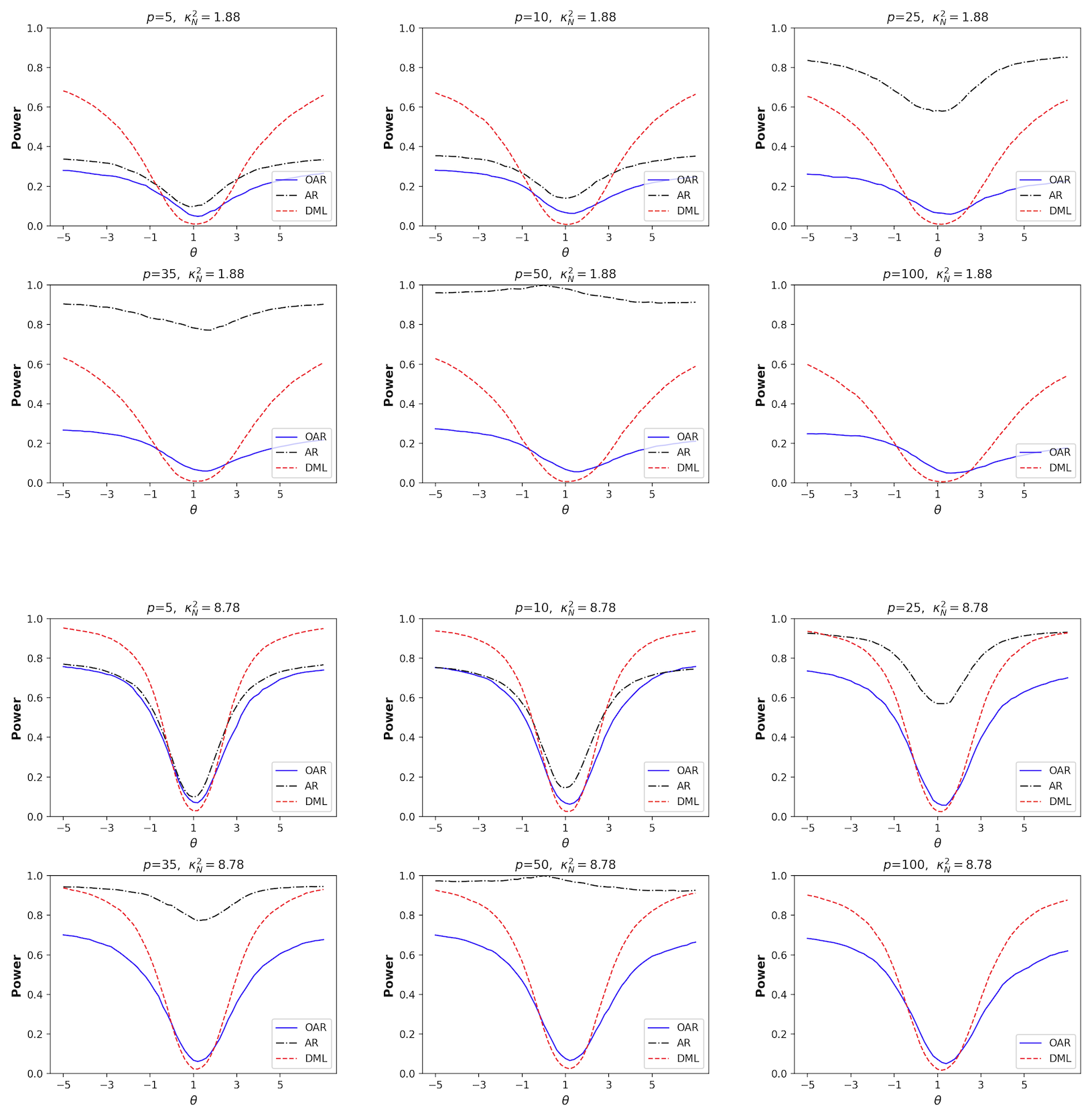}
\caption{Power curves for OAR (solid blue), conventional AR (dash–dot black), and DML (dashed red) for $N=50$ across $p\in\{5,10,25,35,50,100\}$ and instrument strengths $\kappa_N^2\in\{1.88,\,8.87\}$ under the baseline linear–homoskedastic design with $(\rho_Y,\rho_{\sigma},\rho_Z)=(0,0,0)$. The top six panels correspond to $\kappa_N^2=1.88$ and the bottom six to $\kappa_N^2=8.87$. Curves show rejection probabilities for nominal 5\% tests based on 3{,}000 replications.}
\label{figure_N50_power}
\end{figure}

\begin{figure}
\centering
\includegraphics[scale=0.6]{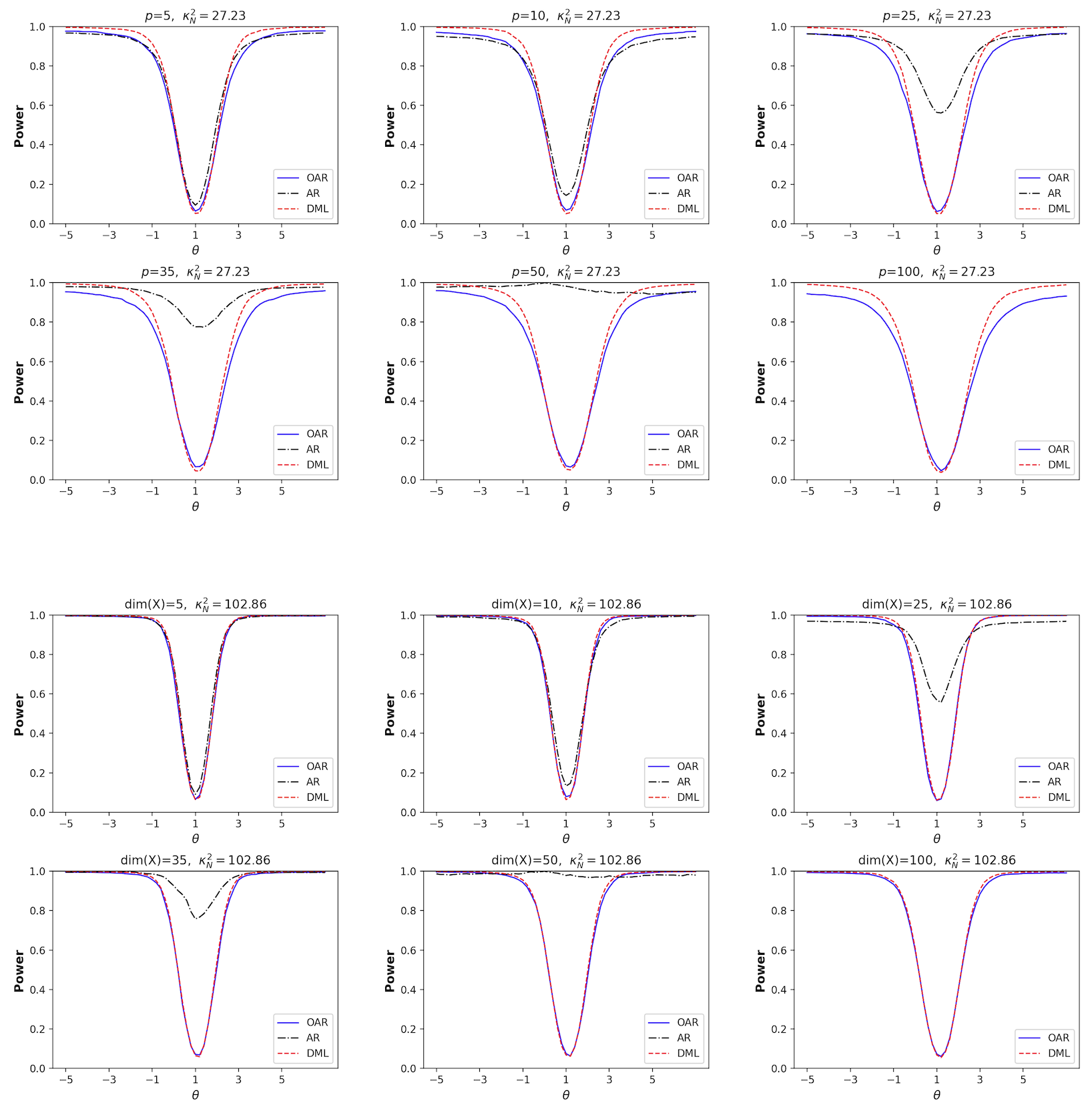}
 \caption{Power curves for OAR (solid blue), conventional AR (dash–dot black), and DML (dashed red) for $N=50$ across $p\in\{5,10,25,35,50,100\}$ and instrument strengths $\kappa_N^2\in\{27.23,\,102.86\}$ under the baseline linear–homoskedastic design with $(\rho_Y,\rho_{\sigma},\rho_Z)=(0,0,0)$. The top six panels correspond to $\kappa_N^2=27.23$ and the bottom six to $\kappa_N^2=102.86$. Curves show rejection probabilities for nominal 5\% tests based on 3{,}000 replications.}
\label{figure_N100_power}
\end{figure}

\begin{table}[htbp]
\centering
\renewcommand{\arraystretch}{1.05}
\setlength{\tabcolsep}{3pt}
\adjustbox{max width=\textwidth}{
\begin{tabular}{l|ccc|ccc|ccc|ccc|ccc|ccc}
\hline\hline 
& \multicolumn{3}{c|}{$p=5$}
& \multicolumn{3}{c|}{$p=10$}
& \multicolumn{3}{c|}{$p=25$}
& \multicolumn{3}{c|}{$p=35$}
& \multicolumn{3}{c|}{$p=50$}
& \multicolumn{3}{c}{$p=100$} \\
\hline
$\rho_Y$      & 0 & 1 & 0 & 0 & 1 & 0 & 0 & 1 & 0 & 0 & 1 & 0 & 0 & 1 & 0 & 0 & 1 & 0 \\
$\rho_\sigma$ & 0 & 0.5 & 0 & 0 & 0.5 & 0 & 0 & 0.5 & 0 & 0 & 0.5 & 0 & 0 & 0.5 & 0 & 0 & 0.5 & 0 \\
$\rho_Z$      & 0 & 0 & 0.5 & 0 & 0 & 0.5 & 0 & 0 & 0.5 & 0 & 0 & 0.5 & 0 & 0 & 0.5 & 0 & 0 & 0.5 \\
\hline
\multicolumn{19}{c}{Panel A: Empirical size ($N=50$, $\kappa_N^2=1.88$)} \\
\hline
OAR & 0.051 & 0.049 & 0.064 & 0.064 & 0.063 & 0.072 & 0.062 & 0.064 & 0.072 & 0.063 & 0.067 & 0.076 & 0.064 & 0.084 & 0.072 & 0.065 & 0.067 & 0.071 \\
AR  & 0.097 & 0.091 & 0.231 & 0.140 & 0.131 & 0.268 & 0.584 & 0.562 & 0.599 & 0.783 & 0.792 & 0.821 & 0.982 & 0.986 & 0.985 & $\times$ & $\times$ & $\times$ \\
DML & 0.008 & 0.005 & 0.005 & 0.006 & 0.006 & 0.007 & 0.008 & 0.008 & 0.008 & 0.009 & 0.006 & 0.009 & 0.006 & 0.007 & 0.009 & 0.007 & 0.007 & 0.005 \\
\hline
\multicolumn{19}{c}{Panel B: Empirical size ($N=50$, $\kappa_N^2=8.87$)} \\
\hline
OAR & 0.070 & 0.057 & 0.074 & 0.067 & 0.063 & 0.069 & 0.066 & 0.065 & 0.075 & 0.067 & 0.064 & 0.074 & 0.075 & 0.075 & 0.081 & 0.071 & 0.079 & 0.087 \\
AR  & 0.097 & 0.093 & 0.231 & 0.144 & 0.129 & 0.267 & 0.571 & 0.558 & 0.599 & 0.782 & 0.787 & 0.818 & 0.977 & 0.987 & 0.986 & $\times$ & $\times$ & $\times$ \\
DML & 0.028 & 0.021 & 0.023 & 0.026 & 0.026 & 0.025 & 0.025 & 0.020 & 0.021 & 0.023 & 0.018 & 0.024 & 0.029 & 0.028 & 0.020 & 0.021 & 0.024 & 0.024 \\
\hline
\multicolumn{19}{c}{Panel C: Empirical size ($N=50$, $\kappa_N^2=27.23$)} \\
\hline
OAR & 0.058 & 0.067 & 0.072 & 0.062 & 0.064 & 0.065 & 0.060 & 0.070 & 0.071 & 0.062 & 0.075 & 0.075 & 0.065 & 0.067 & 0.072 & 0.065 & 0.086 & 0.072 \\
AR  & 0.094 & 0.094 & 0.233 & 0.142 & 0.128 & 0.275 & 0.564 & 0.566 & 0.613 & 0.776 & 0.775 & 0.817 & 0.983 & 0.988 & 0.987 & $\times$ & $\times$ & $\times$ \\
DML & 0.052 & 0.042 & 0.043 & 0.050 & 0.045 & 0.044 & 0.051 & 0.043 & 0.042 & 0.045 & 0.049 & 0.049 & 0.054 & 0.052 & 0.041 & 0.047 & 0.044 & 0.046 \\
\hline
\multicolumn{19}{c}{Panel D: Empirical size ($N=50$, $\kappa_N^2=102.86$)} \\
\hline
OAR & 0.065 & 0.063 & 0.072 & 0.072 & 0.059 & 0.067 & 0.055 & 0.063 & 0.076 & 0.069 & 0.072 & 0.072 & 0.074 & 0.082 & 0.073 & 0.073 & 0.072 & 0.098 \\
AR  & 0.091 & 0.094 & 0.241 & 0.133 & 0.129 & 0.291 & 0.569 & 0.564 & 0.615 & 0.748 & 0.781 & 0.823 & 0.978 & 0.988 & 0.989 & $\times$ & $\times$ & $\times$ \\
DML & 0.064 & 0.057 & 0.057 & 0.061 & 0.052 & 0.061 & 0.064 & 0.055 & 0.062 & 0.064 & 0.052 & 0.059 & 0.067 & 0.057 & 0.064 & 0.068 & 0.060 & 0.057 \\
\hline\hline
\end{tabular}
}
\caption{Empirical size at the 5\% nominal level for $N=50$ under varying identification strength $\kappa_N^2$. Each column corresponds to a DGP defined by $(\rho_Y,\rho_\sigma,\rho_Z)$ at a given $p$. The table reports empirical sizes for the orthogonalized Anderson-Rubin (OAR) test, the conventional AR test, and the DML-based Wald test. Each design uses 3,000 Monte Carlo replications. ``$\times$'' indicates a statistic is not available.}
\label{table_robust}
\end{table}

\section{Empirical Illustrations}
\label{section_application}
\subsection{The Impact of Railroad Access on City Growth}
To demonstrate the methods outlined in the preceding sections, we revisit the instrumental variable analysis by \cite{hornung2015railroads} concerning the impact of railroad access on city growth in 19th-century Prussia. In this study, straight-line corridors between major cities (nodes) are constructed, and whether a city is located on this line is used as an instrument for analysis. We compare the proposed orthogonalized AR test with the conventional AR test and the DML method to assess the effect of railroad access on city population growth.
Our goal is to deepen our understanding of the conclusions presented in the literature. By conducting a new empirical analysis, we keep two econometric considerations in mind: 1. the inclusion of high-dimensional covariates to mitigate unobserved confoundedness, and 2. accounting for the weak identification issue in the data. For the first time, we report confidence intervals that are robust to weak identification and high dimensionality.

Consider the empirical model:
\begin{align*}
&\Ep[D_i|Z_i,X_i]=\frac{\exp (Z_{i} \eta_0 +X_{i}'\beta_2)}{1+\exp(Z_{i} \eta_0 +X_{i}'\beta_2)}, \\
&
\Ep[Y_{it}|D_i,X_i]=D_{i}\theta_0+X_{i}'\beta_1,
\end{align*}
where $Y_{it}$ denotes the urban population growth rate in city $i$ at time period $t$, $D_{i}$ is a dummy variable indicating whether there is railroad access by 1848 in city $i$, and $Z_{i}$ denotes whether city $i$ was located within a straight-line corridor between junction stations (nodes) in 1848. The covariates $X_{i}$ include a lagged dependent variable, distance to the closest node of railroad lines, age composition, primary education of the urban population, county-level concentration of large landholdings, access to main roads, rivers, and ports, pre-railroad city growth from 1831-1837, and the size of the civilian and military population in 1849. 

Within this study, the exclusion restriction condition would be violated if the location of the cities in the straight-line corridor were associated with urban population growth through a channel other than the railroad. The author asserts that the exclusion restriction is satisfied in Hornung (\citeyear{hornung2015railroads}, pg. 714),
\begin{adjustwidth}{1.25cm}{1.25cm}
\textit{When estimating the reduced-form relationship of urban growth on location in the straight-line corridors, we find no correlation with the pre-railroad growth during 1831-1837.}
\end{adjustwidth}
In the context of this study, it is noteworthy that the adoption of railroad technology by cities located on a straight line between two important cities was randomly assigned. This random assignment arises because the positioning of these cities along such lines was not intentionally controlled by any specific entity. In 19th-century Prussia, the government did not dictate railroad construction due to financial limitations. Instead, the decision fell to individual city councils negotiating with private railroad enterprises. Hence, each city had the autonomy to determine whether or not to proceed with railroad construction. Within this study, ``compliers'' refer to (1) cities situated on the straight line between two major cities AND eventually established a railroad station, and (2) cities NOT on such a line AND did NOT get a train station. The second part does not exist in this study as mentioned in Hornung (\citeyear{hornung2015railroads}, pg. 731),
\begin{adjustwidth}{1.25cm}{1.25cm}
\textit{One limitation of using IV estimation approaches lies in the fact that we can only estimate the local average treatment effect of railroad access for cities in the straight-line corridors.}
\end{adjustwidth}
We implement the proposed method on the city-level railroad data from \cite{hornung2015railroads}. As outlined in Table 5 of \cite{hornung2015railroads}, the first-stage F-statistics vary between 26.46 and 38.29. This variation suggested instrument weakness based on the $tF$ critical value function proposed by \cite{lee2022valid}. We conduct a re-analysis by incorporating the polynomial and interaction terms of the original covariates, and present the results in Table \ref{tab:empirical_results}. 

\begin{table}
	\centering
\adjustbox{max width=\textwidth, max height=\textheight}{
		\begin{tabular}{lccccccccccc}
			\hline\hline
			&&(1)&(2)&&(3)&(4)&(5)&(6)&(7)&(8)&(9)\\
			\hline
			$Y_{it}$: population 			&& \multicolumn{2}{c}{Main periods} & \multicolumn{7}{c}{Subperiods}\\
			\cline{3-4}\cline{6-12} growth rate 	
			 &&1831-37 &49-71&& 49-52& 52-55& 55-58&58-61&61-64&64-67&67-71\\
			 
			 \hline \multicolumn{11}{c}{Panel A: AR test} 
			 \\
			\hline
95\% CI && [-0.010,0.013]&[0.007,0.034]&& [-0.009,0.020]&[0.000,0.029]&[0.005,0.050]&[0.001,0.030]&[-0.005,0.041]&[-0.002,0.058]&[0.002,0.049]\\
Length of CI && 0.023&0.027&& 0.029&0.028&0.042&0.029&0.046&0.060&0.047\\

		\hline \multicolumn{11}{c}{Panel B: DML using Lasso }
		\\
		\hline
		LATE && 0.004&0.030&& 0.014&0.024&0.025&0.013&0.038&0.032&0.049\\
95\% CI && [-0.003,0.010]&[0.023,0.036]&& [0.005,0.024]&[0.013,0.035]&[0.014,0.035]&[0.007,0.018]&[0.028,0.049]&[0.020,0.043]&[0.039,0.059]\\
Length of CI && 0.014&0.013&& 0.019&0.011&0.021&0.012&0.021&0.023&0.020\\

					\hline \multicolumn{11}{c}{Panel C: OAR using Lasso }
		\\
		\hline
95\% CI && [-0.020,0.025]&[0.004,0.079]&& [-0.017,0.055]&[0.004,0.042]&[-0.006,0.071]&[-0.005,0.065]&[0.007,0.217]&[-0.021,0.146]&[0.010,0.447]\\
Length of CI && 0.045&0.075&& 0.072&0.038&0.077&0.070&0.211&0.168&0.437\\

						\hline \multicolumn{11}{c}{Panel D: OAR using Ridge }
		\\
		\hline
95\% CI && [-0.016,0.021]&[0.014,0.075]&& [-0.003,0.047]&[0.004,0.046]&[-0.007,0.051]&[-0.001,0.049]&[0.012,0.083]&[-0.005,0.100]&[0.018,0.117]\\
Length of CI && 0.037&0.061&& 0.050&0.043&0.058&0.050&0.071&0.105&0.098\\

						\hline \multicolumn{11}{c}{Panel E: OAR using Elastic Net}
		\\
		\hline
95\% CI && [-0.018,0.022]&[0.010,0.075]&& [-0.008,0.052]&[0.011,0.049]&[0.001,0.064]&[-0.004,0.057]&[0.006,0.126]&[-0.008,0.123]&[0.012,0.169]\\
Length of CI && 0.041&0.065&& 0.060&0.038&0.063&0.060&0.120&0.131&0.157\\

			\hline
			$N$ && 898&906&& 929&924&914&926&924&919&919\\
$p$ && 209&246&& 246&246&246&246&246&246&246\\

			\hline\hline
		 \multicolumn{11}{c}{Panel F: Table 5 in \cite{hornung2015railroads}} 
			 \\
			\hline
95\%			CI && [-0.012,0.012] &[0.009,0.033] &&[0.003,0.027] & [0.005,0.029] &[0.002,0.038] &[0.001,0.021] &[0.007,0.035] &[0.001,0.041] &[0.008,0.036]\\
			Length of CI && 0.024& 0.024 && 0.024&0.024& 0.036&0.020&0.028&0.040&0.028\\
			\hline
				$N$          && 898 & 906&& 929 & 924 &914 &926 &924& 919& 919\\
			 $p$ &&11 & 12&&12&12&12&12&12&12&12\\ 
			\hline\hline
		\end{tabular}
		}
\caption{Reported are LATE point estimates, 95\% confidence intervals (CI), and CI lengths for the effect of railroad access. Panel A reports results from the conventional AR test; Panels B--E report results from DML with Lasso (B) and the proposed OAR with Lasso (C), Ridge (D), and Elastic Net (E). Panel F reports the low-dimensional results from Table 5 of \citet{hornung2015railroads}. Estimates in Panels B--E are computed using 100 splits of 5-fold cross-fitting. Panels A--E use an expanded, high-dimensional set of covariates. Panel F uses the original low-dimensional specification.}
	\label{tab:empirical_results}
\end{table}

\begin{table}
	\centering
\adjustbox{max width=\textwidth, max height=\textheight}{
		\begin{tabular}{lccccccccccc}
			\hline\hline
			&&(1)&(2)&&(3)&(4)&(5)&(6)&(7)&(8)&(9)\\
			\hline
			$Y_{it}$: population 			&& \multicolumn{2}{c}{Main periods} & \multicolumn{7}{c}{Subperiods}\\
			\cline{3-4}\cline{6-12} growth rate 	
			 &&1831-37 &49-71&& 49-52& 52-55& 55-58&58-61&61-64&64-67&67-71\\\hline
			 \multicolumn{11}{c}{Panel A: OAR using Lasso }
\\
\hline
95\% CI        && [-0.016,0.023] &[0.009,0.097] &&[-0.021,0.044] &[0.013,0.060] &[-0.006,0.109] &[-0.006,0.074] &[0.004,0.190] &[-0.015,0.185] &[0.007,0.261]\\
Length of CI && 0.039&0.088&& 0.066&0.047&0.115&0.080&0.187&0.200&0.254\\\hline
			\multicolumn{11}{c}{Panel B: OAR using Ridge (check first regression) }
\\
\hline
95\% CI        && [-0.016,0.019] &[0.016,0.076] &&[-0.001,0.049] &[0.011,0.057] &[0.006,0.068] &[-0.001,0.049] &[0.011,0.093] &[-0.003,0.101] &[0.018,0.126]\\
Length of CI && 0.036&0.060&& 0.049&0.046&0.062&0.050&0.082&0.104&0.108\\
						\hline
						\multicolumn{11}{c}{Panel C: OAR using Elastic Net }
\\
\hline
95\% CI        && [-0.016,0.022]&[0.013,0.085] && [-0.003,0.054]&[0.013,0.057]&[0.002,0.086]&[-0.003,0.066]&[0.009,0.124]&[-0.005,0.127]&[0.014,0.212]\\
Length of CI && 0.038&0.072&& 0.057&0.044&0.084&0.069&0.116&0.131&0.199\\

			\hline
			$N$         && 898 & 906&& 929 & 924 &914 &926 &924& 919& 919\\
			$p$ &&66 & 78&&78&78&78&78&78&78&78\\ 
\hline\hline \multicolumn{11}{c}{Panel D: OAR using Lasso }
\\
\hline
95\% CI        && [-0.017,0.026]&[0.009,0.100]&&[-0.020,0.043]&[0.013,0.058]&[-0.006,0.107]&[-0.006,0.067]&[0.004,0.182]&[-0.015,0.166]&[0.006,0.254]\\
Length of CI && 0.044&0.091&&0.062&0.045&0.113&0.074&0.179&0.180&0.247\\
\hline \multicolumn{11}{c}{Panel E: OAR using Ridge }
\\
\hline
95\% CI        && [-0.020,0.025]&[0.015,0.076] && [-0.002,0.047]&[0.011,0.056]&[0.006,0.067]&[-0.001,0.048]&[0.011,0.088]&[-0.004,0.100]&[0.018,0.124]\\
Length of CI && 0.044&0.061&& 0.049&0.045&0.061&0.049&0.077&0.105&0.106\\
\hline \multicolumn{11}{c}{Panel F: OAR using Elastic Net}
\\
\hline
95\% CI        && [-0.018,0.024]&[0.012,0.082] && [-0.003,0.052]&[0.012,0.055]&[0.002,0.083]&[-0.004,0.063]&[0.008,0.116]&[-0.006,0.124]&[0.012,0.207]\\
Length of CI && 0.042&0.070&& 0.055&0.042&0.081&0.067&0.108&0.130&0.195\\
\hline
$N$         && 898&906&& 929&924&914&926&924&919&919\\
$p$ && 77&90&& 90&90&90&90&90&90&90\\ 
\hline\hline

		\end{tabular}
		}
\caption{LATE point estimates, 95\% CIs, and CI lengths for the effect of railroad access. Panels A--C report OAR results using Lasso (A), Ridge (B), and Elastic Net (C) with $p \in \{66,78\}$ and interaction terms. Panels D--F report OAR results using Lasso (D), Ridge (E), and Elastic Net (F) with $p \in \{77,90\}$ and interaction plus cubic terms. All estimates are based on 100 random splits with 5-fold cross-fitting. Point estimates are the median across splits.}
\label{railroad_different_set}
\end{table}

Table \ref{tab:empirical_results} summarizes the results. To emphasize the robustness of the proposed method within the high-dimensional framework, we report the confidence intervals for the LATE  estimates and the lengths of the confidence intervals.  Panel A reports results from the conventional AR test without a regularization step. Panels B–E report results from DML with Lasso and from the proposed OAR with Lasso, Ridge, and Elastic Net, computed using 5-fold cross-fitting. Panel F reports the low-dimensional results from Table 5 of \citet{hornung2015railroads}.  Panels A--E use an expanded, high-dimensional set of covariates.  Different columns report the results for several dependent variables across diverse time periods. Specifically,  Columns (1) and (2) capture outcomes spanning two main periods: 1831-37, and 1849-1871. Columns (3)-(9) depict findings across seven subperiods.  To mitigate the uncertainty induced by sample splitting, we compute confidence intervals based on the median from 100 repetitions of resampled cross-fitting following \cite{CCDDHNR18}. 

Table \ref{tab:empirical_results} presents a comparison of the three inference methods and illustrates the distinct properties of each approach. The DML estimator relies on a delta-method linearization of the LATE ratio, whose validity requires that the denominator be bounded away from zero. When this condition fails under weak identification, the ratio becomes highly nonlinear, and the Wald approximation produces confidence intervals that are systematically too narrow. Hence, the short DML intervals in Table \ref{tab:empirical_results} do not imply higher precision but rather reflect the lack of robustness of Wald-type inference when the first stage is weak. In contrast, our proposed  orthogonalized AR test  directly tests the moment condition without relying on a linear approximation. It therefore achieves uniform size control irrespective of identification strength. The conventional AR test also has identification-robust properties in low dimensions, but when applied with a large number of controls and no regularization, it suffers from overfitting and instability in covariance estimation, leading to distorted inference. Our orthogonalized AR approach combines the identification-robustness of the AR framework with regularized nuisance estimation and cross-fitting, ensuring valid inference even when the covariate dimension is large relative to the sample size.

Empirically, the results align with these expectations. With the full high-dimensional control set, the OAR confidence intervals are wider than those obtained in DML,  consistent with proper uncertainty quantification when the first stage may be weak, yet still informative. OAR continues to yield statistically significant positive effects in several economically meaningful subperiods (e.g., 1852–55, 1861–64, 1867–71, and for the aggregate 1849–71 period), where the lower bound of the 95\% interval lies above zero. By contrast, a number of subperiod effects that appear significant under DML or under the unregularized AR approach become insignificant under OAR. This indicates that those earlier findings were sensitive to the linear approximation underlying Wald inference or to overfitting in the high-dimensional setting. 
Economically, these results suggest that the strong and precisely estimated effects in the original 2SLS specification in \cite{hornung2015railroads} likely reflected a limited control set that overstated
instrument strength and precision. When model uncertainty and weak identification are incorporated,  the statistical evidence for
a consistently positive effect across subperiods becomes weaker. 
Importantly, OAR produces the same qualitative conclusions across Lasso, Ridge, and Elastic Net regularization, illustrating that our identification-robust results do not hinge on any particular form of penalization.

Table \ref{railroad_different_set} reports OAR results using a reduced set of controls, with the number of covariates $p$ ranging from 66 to 90. The 95\% confidence intervals are very similar to those in Table \ref{tab:empirical_results}. Trimming the control set does not alter the substantive conclusions, which suggests that our findings are not driven by the inclusion of specific covariates, provided the specification remains high dimensional.

\subsection{The Boundary Effects on Rental Prices}
In this subsection, we reexamine the  instrumental variable estimation  by \cite{ambrus2020loss} concerning the long-term consequences of the 1854 cholera outbreak on housing prices. The authors investigate the impacts of  a cholera epidemic in a neighbourhood of 19th-century London on  property values in 1864, a decade following the outbreak.  Page 479 of \cite{ambrus2020loss} presents the background information,
\begin{adjustwidth}{1.25cm}{1.25cm}
\textit{In August 1854, St. James experienced a sudden outbreak of cholera when one of the 13 shallow wells that serviced the parish, the Broad Street pump, became contaminated with cholera bacteria... [So] residents were unaware they should stop using the local water source in order to avoid infection.... [Within] the Broad Street pump (BSP) catchment area, an estimated 16 percent of residents had contracted the disease and approximately 8 percent died.}
\end{adjustwidth}

 In this context, $Y_i$  represents the log rental price of house $i$ in 1864. The variable $D_i$ is an indicator, set to 1 if house $i$  has  at least one cholera death. \cite{ambrus2020loss} use two instruments $Z_i$; we focus on their preferred instrument, which is an indicator for whether the property $i$ falls inside the Broad Street pump (BSP) catchment areas, which were  the primary contaminated area during the  outbreak.   The controls $X_i$ comprise all house characteristic variables listed in Table 1 of \cite{ambrus2020loss}, such as distance to the closest pump, distance to the fire station, distance to the urinal, sewer access, among a total of 26 variables.

%
 In this study,  the ``compliers'' refer to (1)  houses situated within the cholera-affected contaminated areas  AND witnessed at least one cholera-related death, and (2) houses  outside these contaminated zone AND did not experience any cholera fatalities. However, the latter category does not actually exist in the study since the authors limit the sample to properties within a certain distance of the BSP boundary.

\begin{table}
\centering
\adjustbox{max width=\textwidth, max height=\textheight}{
\begin{tabular}{lcccccc}
\hline\hline
 &AR test &DML(Lasso)&OAR(Lasso) &OAR(Ridge)&OAR(Elastic Net)\\
\hline
 \multicolumn{5}{c}{Panel A: $p=52, N=467$}\\\hline
95\% 			CI &[-0.662,-0.383]&[-0.647,0.315] &[-0.740,0.415]&[-0.964,0.243]&[-0.685,0.391]\\
			length of CI &0.279 &0.962 &1.154&1.207&1.077\\
\hline
\multicolumn{5}{c}{Panel B: $p=26, N=467$}\\\hline
95\%			CI &[-0.676,-0.239]&[-0.718,0.344] &[-0.839,0.423]&[-1.038,0.263]&[-0.756,0.388]\\
			length of CI &0.437 &1.061 &1.262&1.301&1.144\\

\hline\hline
\end{tabular}}
\caption{Reported are 95\% CIs and CI lengths for the estimated coefficient on cholera-related deaths obtained from the conventional AR test, DML with Lasso, and the proposed OAR with Lasso, Ridge, and Elastic Net. Inference for the DML and OAR estimators is based on 100 repetitions of resampled cross-fitting with $K=5$ folds. Panel A includes $p=52$ covariates, and Panel B includes $p=26$ covariates.}
\label{table_cholera}
\end{table}

Table B2 in \citet{ambrus2020loss} reports an IV estimate of \(-0.694\) with a 95\% interval \([-1.333,-0.055]\), which excludes zero. The reported first-stage \(F\) statistics are close to 10, which signals possible weak identification under the \(tF\) critical-value function of \citet{lee2022valid}. Table \ref{table_cholera} revisits this application with two control sets, \(p=52\) in Panel A and \(p=26\) in Panel B. In both panels the conventional AR intervals are negative and exclude zero. By contrast, the DML intervals straddle zero, and the OAR intervals also straddle zero and are wider than DML, which is the pattern expected from identification-robust inversion under potentially weak first-stage signal. The OAR conclusions are stable across regularization choices (Lasso, Ridge, and Elastic Net), and they are unchanged when moving from \(p=26\) to \(p=52\). Taken together, the reanalysis indicates that once inference is made identification-robust, the effect of cholera-related deaths on 1864 rents is not statistically distinguishable from zero in our specification, and this conclusion is not driven by the choice of controls or penalty.

Given these insights, we advocate for our proposed orthogonalized AR method in high-dimensional models, especially when weak identification may be a concern.

\section{Conclusion}
\label{section_conclusion}
In this paper, we address the challenge of weak identification within the LATE framework, especially in the presence of high-dimensional covariates. Our primary contribution lies in the introduction of an identification-robust inference method for the high-dimensional LATE framework. This is paired with a user-friendly algorithm for inference and confidence interval construction for the LATE estimate. We validate the uniformly correct asymptotic size of our proposed method.  

There are several potential directions for future research. First, while we rely on the typical microeconometric assumption of i.i.d. sampling, future work could address theoretical developments for data structures exhibiting complex dependence, such as separately and jointly exchangeable arrays. Second, whereas we focused on a single instrument scenario, it would be valuable to develop methods and theories for settings with many weak instruments, as discussed by \cite{mikusheva2022inference, mikusheva2024weak} and \cite{matsushita2024jackknife}. 
We leave these and other extensions for future research.


\newpage
\appendix
\section*{Appendix}
\section{Useful Lemmas}
\label{section_useful_lemmas}
For the convenience of readers, we provide the estimation of penalty loadings in Algorithm \ref{algorithm_lasso} in the following algorithm and the convergence rate for the nuisance parameter in the Lasso logistic and Lasso OLS models in Lemma \ref{lemma_convergence_rate_logistic}. 
\begin{algorithm} \label{algorithm_penalty_loading}(Estimation of Penalty Loadings): Choose a constant $B\geq 1$ as an upper bound on the number of iterations. 
\\
Step 1. Set $b=0$, and for each $j=1,\cdots, p$, initialize $\hat{l}_{1j,0}^k=\hat{l}_{2j,0}^k=\frac{1}{2}\mathbb{E}_{I_k^c}[X_i^2]^{1/2}$ and $\hat{l}_{3j,0}^k=\mathbb{E}_{I_k^c}[X_i^2(Y_i-\mathbb{E}_{I_k^c}[Y_i])^2]^{1/2}$. 
\\
Step 2. Compute the Lasso and partially Lasso estimator, $\hat{\beta}_{11,k}, \hat{\beta}_{12,k}, \hat{\gamma}_k, \hat{\beta}_{21,k}, \hat{\beta}_{22,k},$  based on $\hat{\Psi}^k_1=\text{diag}(\{\hat{l}^k_{1j,0},j=1,\cdots,p\}), \hat{\Psi}^k_2=\text{diag}(\{\hat{l}^k_{2j,0},j=1,\cdots,p\})$, and $\hat{\Psi}^k_3=\text{diag}(\{\hat{l}^k_{3j,0},j=1,\cdots,p\})$. 
\\
Step 3. Set $\hat{l}^k_{1j,b+1}=\mathbb{E}_{I_k^c}[X_i^2(D_i-\Lambda(Z_i\hat\beta_{11,k}+X_i'\hat{\beta}_{12,k}))]^{1/2}$, $\hat{l}^k_{2j,b+1}=\mathbb{E}_{I_k^c}[X_i^2(Z_i-\Lambda(X_i'\hat{\gamma}_k))]^{1/2}$, and $\hat{l}^k_{3j,b+1}=\mathbb{E}_{I_k^c}[X_i^2(Y_i-Z_i\hat{\beta}_{21,k}-X_i'\hat{\beta}_{22,k})^2]^{1/2}$. 
\\
Step 4. If $b>B$, stop; otherwise set $b \leftarrow b+1$ and go to Step 1.
\end{algorithm}
\begin{lemma}(Convergence rate for  Lasso estimators)
\label{lemma_convergence_rate_logistic}
Suppose that Assumption \ref{a:regularity_LATE} holds. In addition, suppose that the penalty choice $\lambda_1^k=\lambda_2^k=\lambda_3^k=1.1\sqrt{|I_k^c|}\Phi^{-1}(1-0.025/p)$.                            Then  with probability $1-o(1)$, 
\begin{align*}
&\|(\hat{\beta}_{11,k},\hat{\beta}_{12,k})-(\beta_{11}^0,\beta_{12}^0)\|_1\vee\|\hat{\gamma}_k-\gamma^0\|_1\vee \|(\hat{\beta}_{21,k},\hat{\beta}_{22,k})-(\beta_{21}^0,\beta_{22}^0)\|_1\lesssim \sqrt{\frac{s_N^2\log(p\vee N)}{N}}, \qquad \text{and}
\\
&\|(\hat{\beta}_{11,k},\hat{\beta}_{12,k})-(\beta_{11}^0,\beta_{12}^0)\|_{2,N}\vee\|\hat{\gamma}_k-\gamma^0\|_{2,N}\vee \|(\hat{\beta}_{21,k},\hat{\beta}_{22,k})-(\beta_{21}^0,\beta_{22}^0)\|_{2,N}\lesssim \sqrt{\frac{s_N\log(p\vee N)}{N}}.
\end{align*}
\end{lemma}

\section{Beyond the LATE Case: General Instrumental Variables Model}
\label{appendix_general}
In this section we extend the main results to a general setting, removing the LATE-style restriction to a single binary endogenous variable and a single binary instrument. In addition, we propose an adjusted method that remains valid in over-identified models. The notation introduced in this section is independent of that used in the paper and is chosen solely to improve readability.
\begin{example}
\label{example_PIIV}(Partially linear IV model with multiple instruments)
Let $W=(Y,D,Z',X')'$, where $Y$ is the  outcome variable, $D$ is a (possibly endogenous) treatment variable, $X\in\mathbb{R}^p$ be a vector of covariates (high- or low-dimensional),  and $Z\in\mathbb{R}^q$ be a vector of instruments. The structural equations are
\begin{align*}
&Y=D\theta_0+g_0(X)+U, \quad \Ep[U|X,Z]=0,
\\
& Z = m_0(X)+V, \quad \Ep[V|X]=0,
\end{align*}
and the orthogonal score is given by $\psi(W;\theta,\eta)=(Y-\ell(X)-\theta(D-r(X)))(Z-m(X))$, where $\theta$ is the target parameter and $\eta=(\ell,m,r)$ collects the nuisance functions. With $\ell_0=\Ep[Y|X]$ and $r_0(X)=\Ep[D|X]$, the score  satisfies the moment condition $\Ep[\psi(W;\theta_0,\eta_0)]=0$ and Neyman orthogonality condition $\partial_{\eta}\Ep[\psi(W;\theta_0,\eta_0)][\eta-\eta_0]=0$.
\end{example}

Since the AR test can suffer power loss in over-identified models as mentioned in \cite{moreira2003conditional}, we develop a procedure that applies to a general instrumental variables model.  Let $(\theta_0,\eta_0)\in \Theta\times \mathcal{T}$ denote the true parameter vector, and restrict attention to scores in the linear Neyman orthogonal form 
\begin{align}
\label{equation_linear_def}
\psi(w;\theta,\eta)=\psi^a(w;\eta)\theta+\psi^b(w;\eta), \text{ for all } w\in\text{supp}(W).
\end{align}


Define the sample moment $q_N(\theta)=N^{-1}\sum_{i=1}^N\psi(W_i;\theta,\eta_0)$ and let $S_N(\cdot)$ denote its expected value, given as $S_N(\cdot)=\Ep[q_N(\cdot)]$. We test $H_0(\theta_0): \Ep[\psi(W;\theta_0,\eta_0)]=0$ against $H_1(\theta_0):\Ep[\psi(W;\theta_0,\eta_0)]\neq 0$.
With this notation,  define the empirical process:
\begin{align}
\mathbb{G}_N(\cdot)=\sqrt{N}q_N(\cdot)-\sqrt{N}S_N(\cdot)=
\frac{1}{\sqrt{N}}\sum_{i=1}^N\left\{ \psi(W_i;\cdot,\eta_0)-\Ep[\psi(W;\cdot,\eta_0)]\right\}.
\end{align}
Later we establish that $\mathbb{G}_N(\cdot)$ weakly converges to a mean-zero Gaussian process $\mathbb{G}(\cdot)$  with a covariance function  $\Omega(\theta_1,\theta_2)=\Ep[\mathbb{G}(\theta_1)\mathbb{G}(\theta_2)']$. To obtain identification-robust inference for high-dimensional instrumental variables models, we adapt the conditional quasi-likelihood-ratio (QLR) test of \cite{andrews2016conditional} to the high-dimensional setting.

For a fixed integer $K>1$, we randomly split the data into $K$ folds, denoted $I_k$ for $k\in[K]$. Using the observations with indices $i\in\{1,\cdots,N\}\setminus I_k$, we  estimate the nuisance parameters $\hat{\eta}_k$ using an appropriate machine learning or nonparametric method. Then we compute certain transformations of  the score  using observations indexed by  $i\in I_k$. Below we  detail the resulting estimators and  the confidence sets  for  high-dimensional  IV framework.

We estimate the empirical process $\mathbb{G}_N(\theta)$ by
\begin{align}
\hat{\mathbb{G}}_N(\theta)=\sqrt{N}\left\{\frac{1}{N}\sum_{k=1}^K\sum_{i\in I_k}\psi(W_i;\theta,\hat\eta_k)-\Ep\left[\psi(W_i;\theta,\hat\eta_k)\right]\right\}.
\end{align}
Here $\hat{\mathbb{G}}_N(\theta)$ is computed on the subsample $i\in I_k$. This computation is repeated $K$ times. An estimator of the sample moment is  $\hat{q}_N(\theta)=N^{-1}\sum_{k=1}^K\sum_{i\in I_k}\psi(W_i;\theta,\hat\eta_k)$.  We propose an estimator of $\Omega(\theta_1,\theta_2)$  as 
 \begin{small}
 \begin{align}
\hat{\Omega}(\theta_1,\theta_2)=\frac{1}{N}\sum_{k=1}^K\sum_{i\in I_k}\psi(W_i;\theta_1,\hat{\eta}_k)\psi(W_i;\theta_2,\hat{\eta}_k)'-\frac{1}{N^2}\sum_{k=1}^K\sum_{k'=1}^K\sum_{i\in I_k, i'\in I_{k'}}\psi(W_i;\theta_1,\hat{\eta}_k)\psi(W_{i'};\theta_2,\hat{\eta}_{k'})'. 
 \end{align}
  \end{small}
Fix a candidate null value $\theta_0\in\Theta$. Following \cite{andrews2016conditional}, define the projected moment process $h_N(\theta)=\sqrt{N}\hat{q}_N(\theta)-\hat{\Omega}(\theta,\theta_0)\hat{\Omega}(\theta_0,\theta_0)^{-1}\sqrt{N}\hat{q}_N(\theta_0)$. Next introduce the bias-corrected distance vector 
\begin{align*}
d_N(\theta)=h_N(\theta)+\hat\Omega(\theta,\theta_0)\hat\Omega(\theta_0,\theta_0)^{-1}\xi_{\theta_0}, \quad \text{with } \quad \xi_{\theta_0}\sim N(0,\hat{\Omega}(\theta_0,\theta_0)),
\end{align*}
 and the associated quadratic form 
 \begin{align*}
 Q_N(\theta)=d_N(\theta)'\hat{\Omega}(\theta,\theta)^{-1}d_N(\theta), \qquad Q_N(\theta_0)=\xi'_{\theta_0}\hat\Omega(\theta_0,\theta_0)^{-1}\xi_{\theta_0}
 \end{align*}
The high-dimensional conditional QLR statistic for testing $H_0(\theta_0): \Ep[\psi(W;\theta_0,\eta_0)]=0$
 is
 \begin{align}
 \label{equation_R_statistics}
R_N(\theta_0,\hat{\Omega})=Q_N(\theta_0)-\inf_{\theta\in\Theta}Q_N(\theta)
 \end{align}
 The statistic in \eqref{equation_R_statistics} is identical to that of \cite{andrews2016conditional}. Our modifications lie in how the projected process $h_N(\theta)$ and the covariance estimator $\hat{\Omega}$ are computed. Specifically, we employ a DML procedure that (i) uses Neyman orthogonal scores to remove first-order regularisation bias and (ii) applies $K$-fold cross-fitting to eliminate own-observation overfitting. 
 
 We compute the level-$\alpha$ critical value $c_\alpha(\theta_0,\hat{\Omega})$ from the conditional null distribution of $R_N(\theta_0,\hat{\Omega})$ using draws of $\xi_{\theta_0}$. The $(1-\alpha)$ confidence set is then $
CI_\alpha=\{\theta\in\Theta:R(\theta,\hat\Omega)\leq c_\alpha(\theta,\hat{\Omega})\},
$ that is, we evaluate $R_N(\cdot)$ and $c_\alpha(\cdot)$ at each candidate value of $\theta$ and retain those not rejected.

  \begin{remark}
  As mentioned in \cite{moreira2003conditional}, this statistic simplifies to the pivotal AR statistic when the dimensionality of the instrument is 1, consistent with our LATE framework in the main text. Nevertheless, we find it pertinent to introduce this conditional test statistic here. This is in light of our aspiration to delve into general TSLS estimation in upcoming research, particularly overidentified cases, which are commonly observed in empirical studies.
  \end{remark}

\subsection{General Asymptotic Behavior of the Orthogonalized QLR Test}
\label{Appendix_subsection_general_theory}
To streamline our discussion, we first standardize the notation. For any finite-dimensional vector $\delta$, we define the $l_1$-norm by $\|\delta\|_1$, $l_2$-norm by $\|\delta\|$, $l_\infty$-norm by $\|\delta\|_\infty$, and $l_0$-seminorm by $\|\delta\|_0$, which represents the number of non-zero components of $\delta$. We define the sample expectation operator  as $\mathbb{E}_N[\cdot]=\frac{1}{N}\sum_{i=1}^N[\cdot]$. The prediction norm of $\delta$ is given by
 $\|x_{ij}'\delta\|_{2,N}=\sqrt{\mathbb{E}_N[(x_{ij}'\delta)^2]}$.   For any matrix $A$, $\|A\|$ denotes the $\ell_2$-norm of the matrix. 
 The sequence $\{\delta_N\}_{N\geq 1}$ consists of positive constants approaching $0$, with the condition that $\delta_N\geq N^{-1/2}$. The sequences $\{a_N\}_{N\geq 1}$, $\{v_N\}_{N\geq 1}$, and $\{K_N\}_{N\geq 1}$ are defined as sets of positive constants, possibly growing to infinity, with $v_N\geq 1$ for all $N\geq 1$. We use $a\lesssim b$ to denote $a\leq cb$ for some $c>0$ that does not depend on $N$. Define the complement of set $B$ by $B^c=\{1,\cdots,N\}\setminus B$.


Let target parameter $\theta\in\Theta$ where $\Theta$ is   a compact subset of $\mathbb{R}^{d_\theta}$, and $\eta\in\mathcal{T}$ for a convex set $\mathcal{T}$. For each sample size $N$, the probability law $P=P_N\in\mathcal{P}_N$ is associated with $W_i$.
The null hypothesis corresponds to  the probability family $\mathcal{P}_0$ of distribution. Throughout we work with bounded-Lipschitz convergence; see Section 1.12 of \cite{van1996weak} for its  equivalence  weak convergence of stochastic processes. Define $J_N = \Ep[\partial_\theta\psi(W;\theta_0,\eta_0)]$ the Jacobian of the score at the true parameters,  $\sigma_{min}(J_N)$ and $\sigma_{max}(J_N)$, the smallest and largest  singular values, and   $\tilde\kappa^2_N=\sqrt{N}\sigma_{min}(J_N)$, the concentration  parameter.
With this notation, we now present the following two assumptions.
\begin{assumption}
\label{assumption_DML1}
For $N\geq 3$ and $P\in\mathcal{P}_N$, the following conditions hold.
\begin{enumerate}[(i)]
\item The score satisfies the moment restriction and Neyman orthogonality condition at the true parameters $(\theta_0,\eta_0)$.
\item The map $\eta\mapsto \Ep[\psi(W;\theta,\eta)]$ is twice continuously G\^ateaux-differentiable on the realization set $\mathcal{T}_N\subset \mathcal{T}$.
\item The score $\psi$ is linear as characterized by (\ref{equation_linear_def}).
\item $\Theta$ is a compact set.
\item $\psi(W;\theta,\eta)$ is continuous with respect to $\theta$.
\item $\{W_i\}_{i=1}^N$ is independent and identically distributed (i.i.d).
\item There exists a finite constant $C$ such that  $\sigma_{max}(J_N)\leq C$. Moreover, $\tilde\kappa^2_N=O(1)$ or $\tilde\kappa^2_N\rightarrow\infty$.
\end{enumerate}
\end{assumption}
\begin{assumption}
\label{assumption_DML2}
For  $N\geq 3$ and $P\in\mathcal{P}_0$, the following conditions hold.
\begin{enumerate}[(i)]
\item Given a random subset $I$ of $\{1,\cdots,N\}$ with size $n=N/K$, the nuisance parameter estimator $\hat{\eta}=\hat{\eta}\left((W_i)_{i\in I^c}\right)$ belongs to the realization set $\mathcal{T}_N$ with probability at least $\Delta_N$, where $\mathcal{T}_N$ contains $\eta_0$ and satisfies the following conditions.
\item Let $ c_1>0, q\geq 4$ be  some finite constants. The following conditions on the rates $m_N, m_N',  r_N'$ hold over $P\in\mathcal{P}_0$:
\begin{align*}
(a) \quad &m_N:=\sup_{\eta\in\mathcal{T}_N}(\Ep[\|\psi(W;\theta,\eta)\|^q])^{1/q}\leq c_1,
\\
(b) \quad & m_N':=\sup_{\eta\in\mathcal{T}_N}(\Ep[\|\psi^a(W;\eta)\|^q])^{1/q}\leq c_1,
\\
(c) \quad
&
r_N':=\sup_{\eta\in\mathcal{T}_N}(\Ep[\|\psi(W;\theta,\eta)-\psi(W;\theta,\eta_0)\|^2])^{1/2}\leq \delta_N.
\end{align*}
\item Let $c_0>0$ be some finite constants. All eigenvalues of the matrix $\Ep[\psi(W;\theta,\eta_0)\psi(W;\theta,\eta_0)']$ are bounded from below by $c_0$.
\end{enumerate}
\end{assumption}
Assumptions \ref{assumption_DML1},  \ref{assumption_DML2} are  related to Assumptions 3.1, 3.2 in \cite{CCDDHNR18}. It is crucial to highlight that 
  Assumption 3.1 (e) in \cite{CCDDHNR18} serves as the strong identification condition in their paper.  We relax this restriction to allow the smallest singular value of the Jacobian to shrink at the local-to-zero rate $\tilde\kappa^2_N=\sqrt{N}\sigma_{min}(J_N)=O(1)$, thereby accommodating weak identification. In the partially linear IV example, it translates into  $\sigma_{min}(\Ep[DV'])=c/\sqrt{N}$ with a fixed constant $c>0$. When $\tilde\kappa^2_N\rightarrow \infty$, the framework also encompasses the strongly identified case, thereby eliminating the need for pre-testing instrument strength.


 Assumption \ref{assumption_DML1} stipulates that the score satisfies the moment condition, Neyman orthogonality condition, and a mild smoothness condition. Assumption \ref{assumption_DML1}(iv)(v) guarantees the compactness of the parameter space $\Theta$ and continuity in $\theta$.
 Assumption \ref{assumption_DML2} introduces some mild regularity conditions. Assumption \ref{assumption_DML2}(i) and (ii) assert that the estimator of the nuisance parameter $\hat{\eta}$ belongs to a shrinking neighbourhood of the true nuisance parameter $\eta_0$ and contracts around $\eta_0$ at a rate of $r_N'$  over $P\in\mathcal{P}_0$. Assumption \ref{assumption_DML2}(iii) ensures a non-degenerate limit distribution.To derive the uniform convergence of the Gaussian process $\hat{\mathbb{G}}_N$, we impose  restrictions over $P\in\mathcal{P}_0$ in Assumption \ref{assumption_DML2}(ii)-(iii).
 
Assumption \ref{assumption_DML2} concerns rates and entropy, and any first-stage learner that satisfies it is admissible. This includes sparsity-regularised or machine learning methods such as lasso, post-lasso, elastic net, boosted trees, random forests, and neural networks when 
$p\gg N$, as well as classical kernel, local-polynomial, or series estimators in low-dimensional settings. Hence, the framework covers both modern high-dimensional and traditional semiparametric cases.

\begin{theorem}
\label{thm_Gaussian_process}
Suppose Assumptions \ref{assumption_DML1} and \ref{assumption_DML2} hold. We have
\begin{align}
\label{equation_gaussian_iid}
\hat{\mathbb{G}}_N(\theta)=\mathbb{G}_N(\theta)+O_P(N^{-1/2}r_N'),
\end{align}
where  recall that
\begin{align*}
&\mathbb{G}_N(\theta)=
\frac{1}{\sqrt{N}}\sum_{i=1}^N\left\{ \psi(W_i;\theta,\eta_0)-\Ep[\psi(W;\theta,\eta_0)]\right\}, \quad \text{ and }
\\
&
\hat{\mathbb{G}}_N(\theta)=\frac{1}{\sqrt{N}}\sum_{k=1}^K\sum_{i\in I_k}\psi(W_i;\theta,\hat\eta_k)-\sqrt{N}\Ep\left[\psi(W_i;\theta,\hat\eta_k)\right].
\end{align*}
The process $\hat{\mathbb{G}}_N(\cdot)$ weakly converges to a centered Gaussian process $\mathbb{G}(\cdot)$ for all $P\in\mathcal{P}_0$ with covariance function $\Omega(\theta_1,\theta_2)=\Ep[\left(\psi(W;\theta_1,\eta_0)-\Ep[\psi(W;\theta_1,\eta_0)]\right)(\psi(W;\theta_2,\eta_0)-\Ep[\psi(W;\theta_2,\eta_0)])]$ as $N$ goes to infinity. Moreover, there is a uniformly consistent variance estimator $\hat{\Omega}(\cdot,\cdot)$ for all $P\in\mathcal{P}_0$ in the form of 
 \begin{align*}
\hat{\Omega}(\theta_1,\theta_2)=\frac{1}{N}\sum_{k=1}^K\sum_{i\in I_k}\psi(W_i;\theta_1,\hat{\eta}_k)\psi(W_i;\theta_2,\hat{\eta}_k)'-\frac{1}{N^2}\sum_{k=1}^K\sum_{k'=1}^K\sum_{i\in I_k, i'\in I_{k'}}\psi(W_i;\theta_1,\hat{\eta}_k)\psi(W_{i'};\theta_2,\hat{\eta}_{k'})',
 \end{align*}
and  for any $\varepsilon>0,$ 
\begin{align*}
\lim_{N\rightarrow\infty}\sup_{P\in\mathcal{P}_0}P\Big\{\sup_{\theta_1,\theta_2}\|\hat{\Omega}(\theta_1,\theta_2)-\Omega(\theta_1,\theta_2)\|>\varepsilon\Big\}=0.
\end{align*}
Under the null, we have 
\begin{align*}
\lim_{N\rightarrow \infty}\sup_{P\in\mathcal{P}_0}P(R_N(\theta_0,\hat{\Omega})>c_\alpha(\theta_0,\hat{\Omega}))\leq \alpha.
\end{align*}
\end{theorem}
\begin{proof}
See Appendix \ref{subsection_proof_thm1}.
\end{proof}

Theorem \ref{thm_Gaussian_process} serves as an extension of   \cite{CCDDHNR18}. In their work,
they  introduce the  pointwise convergence of  the target parameter estimator $\hat\theta$, and discuss the variance estimator associated with the DML estimator $\hat{\theta}$. Our contribution  broadens their results. Specifically,  we demonstrate that our  proposed empirical process exhibits uniform convergence towards a Gaussian process across a broad class of models. Notably,  these models do not  impose restriction on the identification strength, encompassing a wide range of identification scenarios. Moreover,  
our variance estimator $\hat{\Omega}(\theta_1,\theta_2)$ stands for a uniformly consistent estimator for $\Omega(\theta_1,\theta_2)$ under the null. Crucially, the weak convergence result enables us to handle the weak identification challenges effectively.

\begin{lemma} 
\label{lemma_low_condition}
Suppose Assumptions \ref{a:identification_LATE} and \ref{a:regularity_LATE} hold. Then Assumptions \ref{assumption_DML1} and \ref{assumption_DML2} hold for the Neyman orthogonal
score function $\psi(W;\theta,\eta)$ in equation (\ref{equation_score_LATE}) in the LATE framework.
\end{lemma}

\begin{proof}
See Appendix \ref{subsection_proof_thm2}.
\end{proof}

\section{Proofs of the Main Results}
\label{section_proofs}
\subsection{Proof of Theorem \ref{thm_Gaussian_process}}
\label{subsection_proof_thm1}
\begin{proof}
Without loss of generality, we define the size of each fold $I_k$ as $n=N/K$. For notational simplicity, we introduce  the notation $[r]=\{1,\cdots,r\}$ for any $r\in\mathbb{N}$.
Let us  break down the proof into four steps. In Step 1, we demonstrate  equation (\ref{equation_gaussian_iid}) and establish the asymptotic normality of $\hat{\mathbb{G}}_N(\theta)$ over $P\in\mathcal{P}_0$, that is, the asymptotic normality 
of $(\hat{\mathbb{G}}_N(\theta_1),\cdots,\hat{\mathbb{G}}_N(\theta_L))$ for any  $(\theta_1,\cdots,\theta_L)\in\Theta\times\cdots\times\Theta$. In step 2, we  establish the stochastic equicontinuity of $\hat{\mathbb{G}}_N$ over $P\in\mathcal{P}_0$.   Given that $\Theta$ is a compact set, the proof of the weak convergence result is done.
In Step 3, we prove that $\hat{\Omega}(\theta_1,\theta_2)$ serves as  a uniformly consistent estimator for the covariance function $\Omega(\theta_1,\theta_2)$ over $P\in\mathcal{P}_0$. 
In Step 4, we prove the uniformly asymptotically correct size control result. 
\bigskip\\
\noindent \textbf{Step 1.}
In this step, we first establish equation (\ref{equation_gaussian_iid}). Because $K$ is a fixed integer and independent of $N$, it suffices to show that over $P\in\mathcal{P}_0$, for any $k\in[K]$,
\begin{align}
\label{proof_bound_I3}
\Enk[\psi(W;\theta,\hat{\eta}_k)]-\Ep[\psi(W;\theta,\hat{\eta}_k)]-\left(\Enk[\psi(W_i,\theta,\eta_0)]-\Ep[\psi(W;\theta,\eta_0)]\right)=O_p(N^{-1/2}r_N').
\end{align}
For notational simplicity, we define $\Enk[f(W)]=n^{-1}\sum_{i\in I_k}f(W_i)$.
In order to show this, let us fix $k\in[K]$ and introduce an empirical process notation over $P\in\mathcal{P}_0$,
\begin{align*}
\mathbb{G}_{n,k}[\phi(W)]=\frac{1}{\sqrt{n}}\sum_{i\in I_k}(\phi(W_i)-\Ep[\phi(W)]),
\end{align*}
where $\phi$ is any $P_N$-integrable function of $W$. Then by triangle inequality, we have
\begin{align}
&\left\|\Enk[\psi(W;\theta,\hat{\eta}_k)]-\Ep[\psi(W;\theta,\hat{\eta}_k)]-\left(\Enk[\psi(W_i,\theta,\eta_0)]-\Ep[\psi(W;\theta,\eta_0)]\right)\right\|
\\
&
=n^{-1/2} \|\mathbb{G}_{n,k}[\psi(W;\theta,\hat{\eta}_k)]-\mathbb{G}_{n,k}[\psi(W;\theta,\eta_0)]\|:=n^{-1/2}\mathcal{I}_{k3}.
\end{align}
Notice that, conditional on $(W_i)_{i\in I_k^c}$, the estimator $\hat{\eta}_k$ is non-stochastic. Then we have,
\begin{align*}
\Ep[\mathcal{I}_{k3}^2|(W_i)_{i\in I_k^c}]
&=\Ep\left[\|\psi(W;\theta,\hat\eta_k)-\psi(W;\theta_0,\eta_0)\|^2|(W_i)_{i\in I_k^c}\right]
\\
&\leq \sup_{\eta\in\mathcal{T}_N}\Ep\left[\|\psi(W;\theta,\eta)-\psi(W;\theta_0,\eta_0)\|^2|(W_i)_{i\in I_k^c}\right]
\\
&
\leq\sup_{\eta\in\mathcal{T}_N}\Ep[\|\psi(W;\theta,\eta)-\psi(W;\theta_0,\eta_0)\|^2]\leq (r_N')^2.
\end{align*}
This completes the proof of equation (\ref{equation_gaussian_iid}). Combining (\ref{equation_gaussian_iid}) with the Lindeberg-Feller central limit theorem  and the Cramer-Wold device yields the asymptotic normality of $\hat{\mathbb{G}}_{N}(\theta)$ for any $P\in\mathcal{P}_0$.
\bigskip\\
\noindent \textbf{Step 2.}
In this step, we prove the stochastic equicontinuity of $\hat{\mathbb{G}}_N$ over $P\in\mathcal{P}_0$. The stochastic equicontinuity of $\hat{\mathbb{G}}_N$ can be stated as, for any $\epsilon_1>0$, and any $\theta_1,\theta_2\in\Theta$ such that $\|\theta_1-\theta_2\|\leq \delta$,
\begin{align}
\lim _{\delta\rightarrow 0} 
\underset{ N\rightarrow \infty}  {\lim\sup} P\left(\|\hat{\mathbb{G}}_N(\theta_1)-\hat{\mathbb{G}}_N(\theta_2)\|>\epsilon_1\right)=0.
\end{align}
By Markov's  inequality, for any $\epsilon_1>0$,
\begin{align*}
P\left(\left\|\hat{\mathbb{G}}_N(\theta_1)-\hat{\mathbb{G}}_N(\theta_2)\right\|>\epsilon_1\right)\leq \frac{1}{\epsilon_1}\Ep\left[\left\|\hat{\mathbb{G}}_N(\theta_1)-\hat{\mathbb{G}}_N(\theta_2)\right\|\right].
\end{align*}
Thus, it suffices to show that  for each $k\in [K]$.
\begin{align}
\label{equation_asymp_equicont}
\lim _{\delta\rightarrow 0} 
\underset{ N\rightarrow \infty}  {\lim }\sup_{P\in\mathcal{P}_0}\sqrt{N}\Ep\left[\|\Enk[(\theta_1-\theta_2)\psi^a(W;\hat{\eta}_k)]-\Ep[(\theta_1-\theta_2)\psi^a(W;\hat{\eta}_k)]\|\right]=0.\end{align}
Note that 
\begin{align*}
\E_P\left[\big\|\Enk[(\theta_1-\theta_2)\psi^a(W;\hat{\eta}_k)]-\Ep[(\theta_1-\theta_2)\psi^a(W;\hat{\eta}_k)]\big\|^2\right]\leq n^{-1}\delta^2\Ep\left[\|\psi^a(W;\hat{\eta}_k)\|^2\right]\leq n^{-1}\delta^2c_1^2,
\end{align*}
which implies the equation (\ref{equation_asymp_equicont}). Thus, we complete the proof of the asymptotic equicontinuity of $\hat{\mathbb{G}}_N$ over $P\in\mathcal{P}_0$.
\bigskip\\
\noindent \textbf{Step 3.}
In this step, we first show $\hat{\Omega}(\theta_1,\theta_2)=\Omega(\theta_1,\theta_2)+O_P(\rho_N)$, and then we show $\hat{\Omega}$ is a uniformly consistent estimator for $\Omega$ over $P\in\mathcal{P}_0$. To prove the first part,
it suffices to show that over $P\in\mathcal{P}_0$ and each $k\in[K]$,
\begin{align*}
&\mathcal{I}_k=\left\|\Enk[\psi(W;\theta_1,\hat{\eta}_k)\psi(W;\theta_2,\hat{\eta}_k)']-\Ep[\psi(W;\theta_1,\eta_0)\psi(W;\theta_2,\eta_0)']\right\|=O_p(\rho_N), \quad \text{and} 
\\
&
\mathcal{I}_{k}'=\left\|\Enk[\psi(W;\theta,\hat{\eta}_k)]-\Ep[\psi(W;\theta,\eta_0)]\right\|=O_p(\rho_N).
\end{align*}
Note that by triangle inequality, we have 
$
\mathcal{I}_k\leq \mathcal{I}_{k1}+\mathcal{I}_{k2},
$ and 
$
\mathcal{I}_k'\leq \mathcal{I}_{k4}+\mathcal{I}_{k5},
$
where
\begin{align*}
&\mathcal{I}_{k1}=\left\|\Enk[\psi(W;\theta_1,\hat{\eta}_k)\psi(W;\theta_2,\hat{\eta}_k)']-\Enk[\psi(W;\theta_1,\eta_0)\psi(W;\theta_2,\eta_0)']\right\|,
\\
&
\mathcal{I}_{k2}=\left\|\Enk[\psi(W;\theta_1,\eta_0)\psi(W;\theta_2,\eta_0)']-\Ep[\psi(W;\theta_1,\eta_0)\psi(W;\theta_2,\eta_0)']\right\|,
\\
&\mathcal{I}_{k4}=\left\|\Enk[\psi(W;\theta,\hat{\eta}_k)]-\Enk[\psi(W;\theta,\eta_0)]\right\|,
\quad
\mathcal{I}_{k5}=\left\|\Enk[\psi(W;\theta,\eta_0)]-\Ep[\psi(W;\theta,\eta_0)]\right\|.
\end{align*}
First, we  bound $\mathcal{I}_{k2}$ and $\mathcal{I}_{k5}$. Note that for $q\geq 4$, we have 
\begin{align*}
&\Ep[\mathcal{I}^2_{k2}]\leq \sup_{P\in\mathcal{P}_0}
n^{-1}\Ep[\|\psi(W;\theta,\eta_0)\|^4]\leq n^{-1}c_1^4,
\\
&
\Ep[\mathcal{I}_{k5}^2]\leq \sup_{P\in\mathcal{P}_0}
n^{-1}\Ep[\|\psi(W;\theta,\eta_0)\|^2]\leq n^{-1}c_1^2,
\end{align*}
where the last inequality follows from Assumption \ref{assumption_DML2}(ii) and Jensen's inequality. 
Next, we try to bound $\mathcal{I}_{k1}$.
\begin{align*}
\mathcal{I}_{k1}
&=\big\|\frac{1}{n}\sum_{i\in I_k}\left[\psi(W_i;\theta_1,\hat{\eta}_k)\psi(W_i;\theta_2,\hat{\eta}_k)'-\psi(W_i;\theta_1,\eta_0)\psi(W_i;\theta_2,\eta_0)'\right]\big\|
\\
&\leq \frac{1}{n}\sum_{i\in I_k}\big\|\psi(W_i;\theta_1,\hat{\eta}_k)\psi(W_i;\theta_2,\hat{\eta}_k)'-\psi(W_i;\theta_1,\eta_0)\psi(W_i;\theta_2,\eta_0)'\big\|
\\
&\leq \frac{2}{n}\sum_{i\in I_k}\sup_{\theta\in\Theta}\sup_{\eta\in\mathcal{T}_N}\Big(\left\|\psi(W_i;\theta,\hat{\eta}_k)-\psi(W_i;\theta,\eta_0)\right\|\times\|\psi(W_i;\theta,\eta)\|\Big)
\\
&
\leq 
\frac{2}{n}\sum_{i\in I_k}\Big(\sup_{\theta\in\Theta}\left\|\psi(W_i;\theta,\hat{\eta}_k)-\psi(W_i;\theta,\eta_0)\right\|^2\Big)^{1/2}\times\Big(\sup_{\theta\in\Theta}\sup_{\eta\in\mathcal{T}_N}\frac{2}{n}\sum_{i\in I_k}\|\psi(W_i;\theta,\eta)\|^2\Big)^{1/2}
\end{align*}
and the conditional expectation of the first term given $(W_i)_{i\in I_k^c}$ on the event that $\hat{\eta}_k\in\mathcal{T}_N$ is equal to 
\begin{align*}
\sup_{P\in\mathcal{P}_0}\Ep\left[\|\psi(W;\theta,\hat\eta_k)-\psi(W;\theta,\eta_0)\|^2|(W_i)_{i\in I_k^c}\right]\leq \sup_{\eta\in\mathcal{T}_N,P\in\mathcal{P}_0}\Ep\left[\|\psi(W;\theta,\eta)-\psi(W;\theta,\eta_0)\|^2|(W_i)_{i\in I_k^c}\right]=r_N'^2,
\end{align*}
Because the event that $\hat{\eta}_k\in\mathcal{T}_N$ holds with probability $1-\Delta_N=1-o(1)$, it follows that $\mathcal{I}_{k1}=O_P(r_N')=O_P(\delta_N)$. Since $\mathcal{I}_{k2}=O_P(N^{-1/2})$ and $\delta_N\geq N^{-1/2}$, we have $\mathcal{I}_k=O_p(\rho_N)$ with $\rho_N\lesssim \delta_N$. Then we try to bound $\mathcal{I}_{k4}$.
\begin{align*}
\mathcal{I}_{k4}=
&\Big\|\frac{1}{n}\sum_{i\in I_k}\left[\psi(W_i;\theta,\hat{\eta}_k)-\psi(W_i;\theta,\eta_0)\right]\Big\|
\leq \frac{1}{n}\sum_{i\in I_k}\left\|\psi(W_i;\theta,\hat{\eta}_k)-\psi(W_i;\theta,\eta_0)\right\|
\\
&
\leq \sup_{\theta\in\Theta}\left(\frac{1}{n}\sum_{i\in I_k}\left\|\psi(W_i;\theta,\hat{\eta}_k)-\psi(W_i;\theta,\eta_0)\right\|^2\right)^{1/2}.
\end{align*}
By using a similar argument to the one we use to bound $\mathcal{I}_{k1}$, we obtain $\mathcal{I}_{k4}=O_P(r_N')$.  Therefore, we have $\mathcal{I}_{k}'=O_P(\rho_N)$ with $\rho_N\lesssim \delta_N$. This completes the proof of $\hat{\Omega}(\theta_1,\theta_2)=\Omega(\theta_1,\theta_2)+O_P(\rho_N)$. To prove $\hat{\Omega}$ is a uniformly consistent estimator of $\Omega$ over $P\in\mathcal{P}_0$,   we need to show that for any $\varepsilon_2>0$, and any $\theta_1,\theta_2,\theta_1',\theta_2'\in\Theta$ such that $\|\theta_1-\theta_1'\|\leq \delta_1$ and $\|\theta_2-\theta_2'\|\leq \delta_2$, we have
\begin{align*}
\lim _{\delta_1,\delta_2\rightarrow 0} 
\underset{ N\rightarrow \infty}  {\lim } \sup_{P\in\mathcal{P}_0}P\left(\|\hat{\Omega}(\theta_1,\theta_2)-\hat{\Omega}(\theta_1',\theta_2')\|>\epsilon_2\right)=0.
\end{align*} 
By Markov's inequality, for any $\varepsilon_2>0$, 
\begin{align*}
P\left(\|\hat{\Omega}(\theta_1,\theta_2)-\hat{\Omega}(\theta_1',\theta_2')\|>\epsilon_2\right)\leq \frac{1}{\varepsilon_2}\Ep\left[\big\|\hat{\Omega}(\theta_1,\theta_2)-\hat{\Omega}(\theta_1',\theta_2')\big\|\right].
\end{align*}
Thus, it suffices to show that over $P\in\mathcal{P}_0$, for each $k\in[K]$,
\begin{align*}
\mathcal{I}_{k6}:=\Ep\left[\left\|\Enk[\psi(W;\theta_1,\hat{\eta}_k)\psi(W;\theta_2,\hat{\eta}_k)']-\Enk[\psi(W;\theta_1',\hat{\eta}_k)\psi(W;\theta_2',\hat{\eta}_k)']\right\|\right]=0, 
\end{align*}
as $n\rightarrow\infty, \delta_1,\delta_2\rightarrow 0$.
Note that 
\begin{small}
\begin{align*}
\mathcal{I}_{k6}
&\leq\sup_{P\in\mathcal{P}_0}\Ep\left[\En\left[\left\|\psi(W;\theta_1,\hat{\eta}_k)-\psi(W;\theta_1',\hat{\eta}_k)\right\|\cdot\left\|\psi(W;\theta_2,\hat{\eta}_k)\right\|\right]+\En\left[\left\|\psi(W;\theta_2,\hat{\eta}_k)-\psi(W;\theta_2',\hat{\eta}_k)\right\|\cdot\left\|\psi(W;\theta_1',\hat{\eta}_k)\right\|\right]\right]
\\
&
=\sup_{P\in\mathcal{P}_0}\Ep\left[\En\left[\left\|\psi^a(W;\hat{\eta}_k)\cdot(\theta_1-\theta_1')\right\|\cdot\left\|\psi(W;\theta_2,\hat{\eta}_k)\right\|\right]\right]+\Ep\left[\En\left[\left\|\psi^a(W;\hat{\eta}_k)\cdot(\theta_2-\theta_2')\right\|\cdot\left\|\psi(W;\theta_1',\hat{\eta}_k)\right\|\right]\right]
\\
&\leq \sup_{P\in\mathcal{P}_0}(\Ep\left[\|\psi^a(W;\hat{\eta}_k)\|^2\cdot\|\theta_1-\theta_1'\|^2\right])^{1/2}\cdot\left(\Ep\left[\|\psi(W;\theta_2,\hat{\eta}_k)\|^2\right]\right)^{1/2}
\\&
+
\sup_{P\in\mathcal{P}_0}(\Ep\left[\|\psi^a(W;\hat{\eta}_k)\|^2\cdot\|\theta_2-\theta_2'\|^2\right])^{1/2}\cdot\left(\Ep\left[\|\psi(W;\theta_1',\hat{\eta}_k)\|^2\right]\right)^{1/2}
\\
&
\leq \delta_1\sup_{P\in\mathcal{P}_0}(\|\Ep\left[\psi^a(W;\hat{\eta}_k)\|^4\right])^{1/4}\cdot\left(\Ep\left[\|\psi(W;\theta_2,\hat{\eta}_k)\|^4\right]\right)^{1/4}+\delta_2\sup_{P\in\mathcal{P}_0}
(\Ep\left[\|\psi^a(W;\hat{\eta}_k)\|^4\right])^{1/4}\cdot\left(\Ep\left[\|\psi(W;\theta_1',\hat{\eta}_k)\|^4\right]\right)^{1/4}
\\
&
\leq (\delta_1+\delta_2)c_1^2,
\end{align*}
\end{small}
where the second inequality follows from Cauchy-Schwarz inequality, the third inequality follows from Jensen's inequality,  and the last one is from Assumption \ref{assumption_DML2}(ii).
It is obvious that $\lim_{\delta_1,\delta_2\rightarrow 0}\mathcal{I}_{k6}=0$. Therefore, $\hat{\Omega}$ is a uniformly consistent estimator of $\Omega$ over $P\in\mathcal{P}_0$. 
\bigskip\\
\noindent \textbf{Step 4.}
The uniform asymptotic size control relies on Theorem 1 in \cite{andrews2016conditional}. As long as we show that Assumptions 1–4 in \cite{andrews2016conditional} hold, the proof of Theorem \ref{thm_Gaussian_process} is complete.
First, note that Steps 1–3 confirm that Assumptions 1 and 3 in \cite{andrews2016conditional} are satisfied. Next, our Assumption \ref{assumption_DML2} guarantees that Assumption 2 in \cite{andrews2016conditional} is satisfied. Furthermore, Assumption \ref{assumption_DML2}(ii) ensures uniform boundedness of $\Omega(\theta,\theta)$. Lastly, Assumption 4 is directly validated via Lemma 1 in \cite{andrews2016conditional}. Consequently, we complete the proof of Theorem \ref{thm_correct_size}, and thereby the full proof of Theorem \ref{thm_Gaussian_process}.

\end{proof}

\subsection{Proof of Theorem \ref{theorem_weak_convergence}}
\label{subsection_proof_thm2}
\begin{proof}
 As long as we show  Lemma \ref{lemma_low_condition} holds, the proof of Theorem \ref{theorem_weak_convergence}  is done. Let us define $\mathcal{T}_N$ as the set of all $\eta=(g,m,p)$ consisting of $P$-square-integrable function $g, m$ and $p$ such that 
\begin{align*}
&
\|\eta-\eta_0\|_{P,q}\leq c_1,\quad
\|\eta-\eta_0\|_{P,2}\leq \delta_N.
\end{align*}

We proceed in four steps.
\bigskip\\
\noindent \textbf{Step 1.} We first verify the Assumption \ref{assumption_DML1} that the AR-type Neyman orthogonal LATE score  in  (\ref{equation_score_LATE}) satisfies the moment condition  and the Neyman orthogonality condition. It can be easily verified that the moment condition is satisfied. The G\^ateaux derivative in the direction $\eta-\eta_0=(g-g_0,m-m_0,p-p_0)$ is given by 
\begin{align*}
&\partial_{\eta}\Ep\left[\psi(W;\theta_0,\eta)\right]\Big|_{\eta=\eta_0}(\eta-\eta_0)
\\
&=\Ep\left[\left(1-\frac{Z}{p_0(X)}\right)(g(1,X)-g_0(1,X))\right]
-\Ep\left[\left(1-\frac{1-Z}{1-p_0(X)}\right)(g(0,X)-g_0(0,X))\right]
\\
&
-\theta_0\Ep\left[\left(1-\frac{Z}{p_0(X)}\right)(m(1,X)-m_0(1,X))\right]+\theta_0\Ep\left[\left(1-\frac{1-Z}{1-p_0(X)}\right)(m(0,X)-m_0(0,X))\right]
\\
&
+\Ep\Big[\left(\frac{\theta_0Z(D-m_0(1,X))-Z(Y-g_0(1,X))}{p_0(X)^2}+\frac{\theta_0(1-Z)(D-m_0(0,X))-(1-Z)(Y-g_0(0,X))}{(1-p_0(X))^2}\right)
\\
&\times(p(X)-p_0(X))\Big]
\\
&
=0,
\end{align*} 
where the last equality follows from the law of iterated expectations and 
\begin{align}
\label{equation_LIE}
&\Ep[Z|X]=p_0(X),\quad \Ep[Z(Y-g_0(1,X))|X,Z]=0,\quad \Ep[Z(D-m_0(1,X))|X,Z]=0,
\\ \nonumber
&\Ep[1-Z|X]=1-p_0(X),\quad \Ep[(1-Z)(Y-g_0(0,X))|X,Z]=0,\quad \Ep[(1-Z)(D-m_0(0,X))|X,Z]=0.
\end{align}
Referring to the  definitions of the AR-type score for the LATE in  (\ref{equation_score_LATE}) and linear orthogonal score in equation \eqref{equation_linear_def}, we have
\begin{align*}
&\psi^b(W;\eta)=g(1,X)-g(0,X)+\frac{Z(Y-g(1,X))}{p(X)}-\frac{(1-Z)(Y-g(0,X))}{1-p(X)},
\\
&\psi^a(W;\eta)=
-m(1,X)+m(0,X)-\frac{Z(D-m(1,X))}{p(X)}+\frac{(1-Z)(D-m(0,X))}{1-p(X)}.
\end{align*}
Then we have $\psi(W;\theta,\eta)=\psi^b(W;\eta)+\theta\times\psi^a(W;\eta)$.
 Therefore, all the conditions in  Assumption \ref{assumption_DML1} hold.
\bigskip\\
\noindent \textbf{Step 2.}
Next, let us verify Assumption \ref{assumption_DML2}(iii). 
Note that 
\begin{align*}
&\Ep\left[\psi(W;\theta,\eta_0)^2\right]
=\Ep\big[\left(g_0(1,X)-g_0(0,X)-\theta(m_0(1,X)-m_0(0,X))\right)^2\big]
\\
&+\Ep\Big[\Big(\frac{Z(Y-g_0(1,X))}{p_0(X)}-\frac{(1-Z)(Y-g_0(0,X))}{1-p_0(X)}-\theta\big(\frac{Z(D-m_0(1,X))}{p_0(X)}-\frac{(1-Z)(D-m_0(0,X))}{1-p_0(X)}\big)\Big)^2\Big]
\\
&
\geq 
\Ep\Big[\Big(\frac{Z(Y-g_0(1,X))}{p_0(X)}-\frac{(1-Z)(Y-g_0(0,X))}{1-p_0(X)}\Big)^2\Big]-\theta^2\Ep\Big[\Big(\frac{Z(D-m_0(1,X))}{p_0(X)}-\frac{(1-Z)(D-m_0(0,X))}{1-p_0(X)}\Big)^2\Big]
\\
&
\geq
\Ep\left[\frac{Z^2(Y-g_0(1,X))^2}{p_0(X)^2}\right]+\Ep\left[\frac{(1-Z)^2(Y-g_0(0,X))^2}{(1-p_0(X))^2}\right]
\\
&
\geq 
\frac{\Ep\left[Z(Y-g_0(1,X))^2+(1-Z)(Y-g_0(0,X))^2\right]}{(1-\varepsilon)^2}
\\
&
=\frac{\Ep[u^2]}{(1-\varepsilon)^2}
\geq 
\frac{c_0^2}{(1-\varepsilon)^2},
\end{align*}
where the first equality holds since the interaction term equals zero by the equations in \eqref{equation_LIE}, the third inequality follows from the facts that $\varepsilon\leq p_0(X)\leq 1-\varepsilon$, and the last equality follows from Assumption \ref{a:regularity_LATE}(iii). Thus Assumption \ref{assumption_DML2}(iii) is satisfied.
\bigskip\\
\noindent \textbf{Step 3.}
Next, we verify Assumption \ref{assumption_DML2}(i). By Lemmas \ref{lemma_convergence_rate_logistic}  invoked  by Assumption \ref{a:regularity_LATE}, with probability $1-o(1)$,
\begin{align*}
\|\hat{\eta}-\eta_0\|_1\lesssim \sqrt{\frac{s_N^2\log(p\vee N)}{N}}, \quad \text{and} \quad \|\hat{\eta}-\eta_0\|_{2,N}\lesssim \sqrt{\frac{s_N\log(p\vee N)}{N}}.
\end{align*}
Thus Assumption \ref{assumption_DML2}(i) is satisfied.

\noindent \textbf{Step 4.}
\label{step_3}
Next, let us verify the condition in Assumption \ref{assumption_DML2}(ii). 
Note that 
\begin{align*}
\|g_0(D,X)\|_{P,q}
&=(\Ep[|g_0(D,X)|^q])^{1/q}
\\
&
\geq  (\Ep[|g_0(1,X)|^q P(D=1|X)+|g_0(0,X)|^q P(D=0|X)])^{1/q}
\\ 
&
\geq \varepsilon^{1/q}(\Ep[|g_0(1,X)|^q]+\Ep[|g_0(0,X)|^q])^{1/q}
\\
&
\geq \varepsilon^{1/q}(\Ep[|g_0(1,X)|^q]\vee\Ep[|g_0(0,X)|^q])^{1/q}
\\
&
\geq \varepsilon^{1/q}(\|g_0(1,X)\|_{P,q}\vee\|g_0(0,X)\|_{P,q}).
\end{align*}
Since $\|g_0(D,X)\|_{P,q}\leq \|Y\|_{P,q}\leq c_1$ by Assumption \ref{a:regularity_LATE}(vi), we have 
\begin{align*}
\|g_0(1,X)\|_{P,q}\leq c_1/\varepsilon^{1/q}, \text{and} \|g_0(0,X)\|_{P,q}\leq c_1/\varepsilon^{1/q}.
\end{align*}
By using similar arguments,  we obtain 
\begin{align}
\label{basic_inequality}
&\|g(1,X)-g_0(1,X)\|_{P,q}
\leq c_1/\varepsilon^{1/q}, \quad \|g(0,X)-g_0(0,X)\|_{P,q}\leq c_1/\varepsilon^{1/q},
\\ \nonumber
&\|m_0(1,X)\|_{P,q}
\leq 1/\varepsilon^{1/q}, \quad \|m_0(0,X)\|_{P,q}\leq 1/\varepsilon^{1/q},
\\ \nonumber
&\|m(1,X)-m_0(1,X)\|_{P,q}
\leq c_1/\varepsilon^{1/q}, \quad \|m(0,X)-m_0(0,X)\|_{P,q}\leq c_1/\varepsilon^{1/q},
\end{align}
since $\|m_0(D,X)\|_{P,q}\leq 1$, $\|g(D,X)-g_0(D,X)\|_{P,q}\leq c_1$, and $\|m(Z,X)-m_0(Z,X)\|_{P,q}\leq c_1$.
By calculation, we obtain 
\begin{align*}
&\|\psi^a(W;\eta)\|_{P,q}
\leq (1+\varepsilon^{-1})(\|m(1,X)\|_{P,q}+\|m(0,X)\|_{P,q})+2/\varepsilon
\\
& \leq (1+\varepsilon^{-1})(\|m(1,X)-m_0(1,X)\|_{P,q}+\|m_0(1,X)\|_{P,q}+\|m(0,X)-m_0(0,X)\|_{P,q}+\|m_0(0,X)\|_{P,q})+2/\varepsilon
\\
&
\leq (1+\varepsilon^{-1})(2c_1\varepsilon^{-1/q}+2\varepsilon^{-1/q})+2\varepsilon^{-1}:=c_{\varepsilon1},
\\
&
\|\psi^b(W;\eta)\|_{P,q}
\leq (1+\varepsilon^{-1})(\|g(1,X)\|_{P,q}+\|g(0,X)\|_{P,q})+2\|Y\|_{P,q}/\varepsilon
\\
&\leq (1+\varepsilon^{-1})(2c_1\varepsilon^{-1/q}+2\varepsilon^{-1/q})+2c_1\varepsilon^{-1}:= c_{\varepsilon2},
\end{align*}
where $c_{\varepsilon1}$ and $c_{\varepsilon2}$ are constants related to $\varepsilon$ instead of $N$. Note that this completes the verification of Assumption \ref{assumption_DML2} (b)
Therefore, under the null, we have
\begin{align*}
&(\Ep[\|\psi(W;\theta,\eta)\|^q])^{1/q}
=\|\psi(W;\theta,\eta)\|_{P,q}\leq \|\psi(W;\theta,\eta)-\psi(W;\theta_0,\eta)\|_{P,q}+\|\psi(W;\theta_0,\eta)\|_{P,q}
\\
&
\leq |\theta-\theta_0|\times \|\psi^a(W;\eta)\|_{P,q}+\|\psi^b(W;\eta)\|_{P,q}
+|\theta_0|\times\|\psi^a(W;\eta)\|_{P,q}
\\
&
\leq |\theta-\theta_0|c_{\varepsilon1}+c_{\varepsilon2}+|\theta_0|c_{\varepsilon1}
\lesssim 1,
\end{align*}
where the last inequality needs the assumption that $\Theta$ is a compact set by Assumption \ref{a:regularity_LATE}(v). This completes the verification of  Assumption \ref{assumption_DML2}(ii)(a).
 
Next, let us verify the condition in  Assumption \ref{assumption_DML2}(ii)(c). For any $\eta=(g,m,p)$, by the triangle inequality, 
\begin{align*}
(\Ep[\|\psi(W;\theta,\eta)-\psi(W;\theta,\eta_0)\|^2])^{1/2}=\|\psi(W;\theta,\eta)-\psi(W;\theta,\eta_0)\|_{P,2}\leq \mathcal{I}_1+\mathcal{I}_2+\mathcal{I}_3+\mathcal{I}_4,
\end{align*}
where 
\begin{align*}
&\mathcal{I}_1:=\|g(1,X)-g_0(1,X)\|_{P,2}+\|g(0,X)-g_0(0,X)\|_{P,2},
\\
&\mathcal{I}_2:=|\theta|\times\left(\|m(1,X)-m_0(1,X)\|_{P,2}+\|m(0,X)-m_0(0,X)\|_{P,2}\right),
\\
&
\mathcal{I}_3:=\left\|\frac{Z(Y-g(1,X))}{p(X)}-\frac{Z(Y-g_0(1,X))}{p_0(X)}\right\|_{P,2}+\left\|\frac{(1-Z)(Y-g(0,X))}{1-p(X)}-\frac{(1-Z)(Y-g_0(0,X))}{1-p_0(X)}\right\|_{P,2},
\\
&
\mathcal{I}_4:=|\theta|\times\left(\left\|\frac{Z(D-m(1,X))}{p(X)}-\frac{Z(D-m_0(1,X))}{p_0(X)}\right\|_{P,2}+\left\|\frac{(1-Z)(D-m(0,X))}{1-p(X)}-\frac{(1-Z)(D-m_0(0,X))}{1-p_0(X)}\right\|_{P,2}\right).
\end{align*}
By using a similar argument to the one in obtaining equation \eqref{basic_inequality}, we have
\begin{align*}
&\|g(1,X)-g_0(1,X)\|_{P,2}\leq \delta_N/\varepsilon^{1/q}, \quad \|g(0,X)-g_0(0,X)\|_{P,2}\leq \delta_N/\varepsilon^{1/q},
\\
&
\|m(1,X)-m_0(1,X)\|_{P,2}\leq \delta_N/\varepsilon^{1/q}, \quad \|m(0,X)-m_0(0,X)\|_{P,2}\leq \delta_N/\varepsilon^{1/q}.
\end{align*}
so $\mathcal{I}_1\leq 2\delta_N/\varepsilon^{1/q} $ and $\mathcal{I}_2\lesssim 2\delta_N/\varepsilon^{1/q}$. To bound $\mathcal{I}_3$, we have 
\begin{align*}
\mathcal{I}_3
&\leq \varepsilon^{-2}\times\Big(\|Zp_0(X)(Y-g(1,X))-Zp(X)(Y-g_0(1,X))\|_{P,2}
\\
&
+\|(1-Z)(1-p_0(X))(Y-g(0,X))-(1-Z)(1-p(X))(Y-g_0(0,X))\|_{P,2}\Big)
\\
&
\leq
\varepsilon^{-2}\times
\Big(\|p_0(X)(u+g_0(1,X)-g(1,X))-p(X)u\|_{P,2}
\\
&
+\|(1-p_0(X))(u+g_0(0,X)-g(0,X))-(1-p(X))u\|_{P,2}\Big)
\\
&
\leq
\varepsilon^{-2}\times
\Big(\|p_0(X)(g_0(1,X)-g(1,X))\|_{P,2}+\|(p(X)-p_0(X))u\|_{P,2}
\\
&
+\|(1-p_0(X))(g_0(0,X)-g(0,X))\|+\|(p(X)-p_0(X))u\|_{P,2}\Big)
\\
&
\leq 
\varepsilon^{-2}\times
\Big(\|(g_0(1,X)-g(1,X))\|_{P,2}+\sqrt{c_1}\|p(X)-p_0(X)\|_{P,2}
\\
&
+\|(g_0(0,X)-g(0,X))\|+\sqrt{c_1}\|p(X)-p_0(X)\|_{P,2}\Big)
\\
&
\leq 
\varepsilon^{-2}\times (2/\varepsilon^{1/q}+2\sqrt{c_1})\delta_N \leq c_{\varepsilon_3} \delta_N, 
\end{align*}
where $c_{\varepsilon_3}\geq \varepsilon^{-2}\times (2/\varepsilon^{1/q}+2\sqrt{c_1})$, 
 the first inequality follows from $\varepsilon\leq p(X)\leq 1-\varepsilon$ and $\varepsilon\leq 1-p(X)\leq 1-\varepsilon$, and the fourth one follows from Assumption \ref{a:regularity_LATE}(vi). We use a similar argument to bound $\mathcal{I}_4$ and obtain that $\mathcal{I}_4\lesssim \delta_N$. Therefore, we have $\|\psi(W;\theta,\eta)-\psi(W;\theta,\eta_0)\|_{P,2}\lesssim \delta_N$, which completes the verification of  Assumption \ref{assumption_DML2}(ii).
\end{proof}
\subsection{Proof of Theorem \ref{thm_correct_size}}
\label{subsection_proof_thm3}

\begin{proof}
By Theorem \ref{theorem_weak_convergence},
\[
\lim_{N\to\infty}\ \sup_{P\in\mathcal P_0}
P\!\left(\, \bigl|\,\hat\Omega(\theta_0,\theta_0)-\Omega(\theta_0,\theta_0)\,\bigr|>\varepsilon \right)=0,
\quad \forall\,\varepsilon>0.
\]
The same theorem gives the pointwise CLT at $\theta_0$ uniformly in $P$:
\[
\lim_{N\to\infty}\ \sup_{P\in\mathcal P_0}
\Bigl|\, P\!\Bigl( \sqrt{N}\,\hat q_N(\theta_0) \le x \Bigr)
- \Phi\!\Bigl( x/\sqrt{\Omega(\theta_0,\theta_0)} \Bigr) \Bigr| \ =\ 0,\quad  x\in\mathbb R.
\]
By Slutsky's lemma,
\[
\lim_{N\to\infty}\ \sup_{P\in\mathcal P_0}
\Bigl|\, P\!\Bigl( \tfrac{\sqrt{N}\,\hat q_N(\theta_0)}{\sqrt{\hat\Omega(\theta_0,\theta_0)}} \le z \Bigr)
- \Phi(z) \Bigr| \ =\ 0,\quad \forall z\in\mathbb R,
\]
which implies
\[
\lim_{N\to\infty}\ \sup_{P\in\mathcal P_0}
\Bigl|\, P\!\Bigl( AR(\theta_0) \le t \Bigr)
- F_{\chi^2_1}(t) \Bigr| \ =\ 0,\quad \forall t \text{ a continuity point of } F_{\chi^2_1}.
\]
Choosing $t=\chi^2_{1,1-\alpha}$ yields the stated uniform size. 
\end{proof}

\subsection{Proof of Proposition \ref{prop_power}}
\label{proof_proposition_power}
\begin{proof}
Recall $S(\theta)=\Ep[\psi(W;\theta,\eta_0)]$ and $\partial_\theta S(\theta_0)=-B_N$. 
First, consider the case of weak identification with a fixed alternative.
Assume $\kappa_N^2=O(1)$, i.e. $\sqrt{N}B_N\to \kappa\in(0,\infty)$, and the data are generated under the fixed alternative $\theta=\theta_0+\Delta$ with fixed $\Delta\neq 0$. A first–order expansion of $S$ at $\theta_0$ yields
\[
\Ep[\psi(W;\theta_0,\eta_0)]
= S(\theta_0+\Delta)-\partial_{\theta}S(\theta_0)\,\Delta + o(1)
= -B_N\Delta + o(1).
\]
Hence $\sqrt{N}\,\Ep[\psi(W;\theta_0,\eta_0)]\to \nu=-\kappa\Delta$. Orthogonality and cross-fitting imply that nuisance-estimation error is $o_p(1)$ in $\sqrt{N}\hat{q}_N(\theta_0)$, so
$
\sqrt{N}\hat{q}_N(\theta_0)\xrightarrow{d}\mathcal{N}\!\big(\nu,\Omega(\theta_0,\theta_0)\big).
$
By Slutsky’s theorem and the uniform consistency of $\hat{\Omega}(\theta_0,\theta_0)$, 
\[
AR(\theta_0)\xrightarrow{d}\chi_1^2(\Xi),\qquad \text{with}\quad 
\Xi=\frac{\nu^2}{\Omega(\theta_0,\theta_0)}=\frac{\kappa^2\Delta^2}{\Omega(\theta_0,\theta_0)}.
\]
This shows that the orthogonalized AR test retains nontrivial power under bounded concentration.

Next, consider the case of strong identification with a local alternative. 
Assume $B_N\to B>0$ and the data are generated under the local alternative $\theta=\theta_0+h/\sqrt{N}$ with fixed $h$. Then
\[
\Ep[\psi(W;\theta_0,\eta_0)]= -B_N\frac{h}{\sqrt{N}} + o(N^{-1/2}),
\]
so
$
\sqrt{N}\hat{q}_N(\theta_0)\xrightarrow{d}\mathcal{N}\!\big(-Bh,\Omega(\theta_0,\theta_0)\big)
\quad\text{and}\quad
AR(\theta_0)\xrightarrow{d}\chi_1^2(\Upsilon),
$
with
$
\Upsilon=\frac{B^2h^2}{\Omega(\theta_0,\theta_0)}.
$
Because $\psi$ is the efficient orthogonal score, the semiparametric information is 
$I_{\text{eff}}=B^2/\Omega(\theta_0,\theta_0)$, hence 
$\Upsilon = I_{\text{eff}}\,h^2$. 
This verifies that the orthogonalized AR statistic achieves the semiparametric local power \emph{envelope} under strong identification.
\end{proof}

\subsection{Proof of Lemma \ref{lemma_convergence_rate_logistic} }
\label{subsection_proof_lemma2}

\begin{proof}
The proof follows  from Theorems 6.1 and 6.2 in \cite{belloni2017program}. We verify that the conditions required by these theorems hold in our setting. First, Assumption \ref{a:regularity_LATE}(ii)  implies the restricted eigenvalue condition required in \cite{belloni2017program} and Assumption \ref{a:regularity_LATE}(i) implies the sparsity condition. In addition, Condition 6.1(iii)(iv) and  6.2(iii)(iv) in \cite{belloni2017program} are implied by Assumption \ref{a:regularity_LATE}(iii)--(iv). Therefore, setting the penalty parameters $\lambda_1^k = \lambda_2^k = \lambda_3^k = 1.1\sqrt{|I_k^c|} \Phi^{-1}(1 - 0.025/p)$ ensures that the lemma holds.
\end{proof}

\section{Empirical 95\% Coverage of Confidence Sets}
This appendix summarizes the empirical 95\% coverage performance of all inference procedures examined in Section \ref{section_simulation}. For each design, we simulate 3{,}000 replications and record the proportion of confidence sets that contain the true parameter value.
Table  6 reports coverage results for the proposed OAR test, the conventional AR statistic, and the  DML approach. All results are generated under the baseline linear--homoskedastic design $(\rho_Y,\rho_\sigma,\rho_Z)=(0,0,0)$. The left panel presents results for $N=50$, and the right panel presents results for $N=100$. Across all configurations of identification strength and dimensionality, the OAR procedure maintains empirical coverage close to the nominal 95 percent level, demonstrating robustness to weak first-stage settings. In contrast, the conventional AR test tends to under-cover as dimensionality increases, while the DML approach is overly conservative and over-covers when instruments are weak. These findings reinforce the theoretical size-control properties of the proposed inference method.

\begin{table}[htbp]
\centering
\renewcommand{\arraystretch}{0.75}
\setlength{\tabcolsep}{3pt}
\resizebox{\textwidth}{!}{%
\begin{tabular}{cccc|ccc|cccc|ccc}
\hline\hline
\multicolumn{4}{c|}{\textbf{DGP }} & \multicolumn{3}{c|}{\textbf{95\% Coverage }} &
\multicolumn{4}{c|}{\textbf{DGP}} & \multicolumn{3}{c}{\textbf{95\% Coverage }} \\
\hline
$\kappa_N^2$ & $P_C$ & $N$ & $p$ & OAR & AR & DML &
$\kappa_N^2$ & $P_C$ & $N$ & $p$ & OAR & AR & DML \\
\hline
1.88 & 0.21 & 50  & 5   & 0.961 & 0.947 & 0.991 &   1.88 & 0.15 & 100 & 5   & 0.954 & 0.961 & 0.995 \\
 &   &   & 10  & 0.963 & 0.921 & 0.993 &    &   &   & 10  & 0.956 & 0.958 & 0.995 \\
 &  &   & 25  & 0.961 & 0.430 & 0.992 &    &  & & 25  & 0.954 & 0.929 & 0.995 \\
 &   &   & 35  & 0.961 & 0.217 & 0.991 &    &   &   & 50  & 0.951 & 0.466 & 0.995 \\
 &   &   & 50  & 0.962 & 0.018 & 0.994 &    &   &   & 75  & 0.942 & 0.187 & 0.995 \\
 &   &   & 100 & 0.959 & $\times$ & 0.993 &    &   &   & 100 & 0.941 & 0.013 & 0.996 \\
\hline
8.87 & 0.42  & 50 & 5   & 0.942 & 0.911 & 0.979 &     8.87 & 0.30 & 100 & 5   & 0.932 & 0.932 & 0.983 \\
 &   &   & 10  & 0.938 & 0.872 & 0.974 &    &   &   & 10  & 0.932 & 0.923 & 0.981 \\
 &  &   & 25  & 0.941 & 0.432 & 0.974 &    &  & & 25  & 0.925 & 0.893 & 0.981 \\
 &   &   & 35  & 0.936 & 0.218 & 0.979 &    &   &   & 50  & 0.927 & 0.475 & 0.979 \\
 &   &   & 50  & 0.937 & 0.023 & 0.976 &    &   &   & 75  & 0.929 & 0.187 & 0.977 \\
 &   &   & 100 & 0.935 & $\times$ & 0.979 &    &   &   & 100 & 0.931 & 0.014 & 0.979 \\
\hline
27.23 & 0.63 & 50 & 5     & 0.937 & 0.907 & 0.948 & 27.23 & 0.45 & 100 & 5   & 0.925 & 0.927 & 0.954 \\
 &   &   & 10  & 0.932 & 0.860 & 0.949 &    &   &   & 10  & 0.931 & 0.916 & 0.951 \\
 &  &   & 25  & 0.942 & 0.436 & 0.949 &    &  &  & 25  & 0.935 & 0.877 & 0.955 \\
 &   &   & 35  & 0.943 & 0.224 & 0.955 &    &   &   & 50  & 0.929 & 0.479 & 0.955 \\
 &   &   & 50  & 0.939 & 0.018 & 0.943 &    &   &   & 75  & 0.931 & 0.203 & 0.952 \\
 &   &   & 100 & 0.935 & $\times$ & 0.953 &    &   &   & 100 & 0.932 & 0.013 & 0.948 \\
\hline
102.86 & 0.84 & 50 & 5   & 0.939 & 0.910 & 0.938 & 102.86 & 0.60 & 100 & 5   & 0.925 & 0.927 & 0.942 \\
 &   &   & 10  & 0.927 & 0.868 & 0.936 &    &   &   & 10  & 0.927 & 0.917 & 0.939 \\
 &  &   & 25  & 0.941 & 0.431 & 0.938 &    &  &  & 25  & 0.921 & 0.873 & 0.942 \\
 &   &   & 35  & 0.934 & 0.242 & 0.936 &    &   &   & 50  & 0.931 & 0.481 & 0.936 \\
 &   &   & 50  & 0.929 & 0.022 & 0.933 &    &   &   & 75  & 0.937 & 0.194 & 0.941 \\
 &   &   & 100 & 0.929 & $\times$ & 0.932 &    &   &   & 100 & 0.921 & 0.011 & 0.937 \\
\hline\hline
\end{tabular}\label{table_95coverage}
}
\caption{Empirical 95\% coverage rates for $N=50$ (left panel) and $N=100$ (right panel) across identification strength and dimensionality under the baseline linear–homoskedastic design with $(\rho_Y,\rho_{\sigma},\rho_Z)=(0,0,0)$. }
\end{table}

\newpage
\bibliography{biblio}
\end{document}